\documentclass[10pt]{article}
\usepackage[utf8]{inputenc}
\usepackage[margin=1in]{geometry}
\usepackage{rotating}
\usepackage{booktabs}
\usepackage{array}
\usepackage{amsmath, amsfonts, amssymb}
\usepackage{color}
\usepackage{hyperref}
\usepackage{cleveref}
\usepackage{bbm}
\usepackage{bbold}
\usepackage{algorithm}
\usepackage[noend]{algpseudocode}
\usepackage{subcaption}
\usepackage{graphbox}
\usepackage{authblk}

\usepackage{mathtools}
\mathtoolsset{showonlyrefs}

\usepackage[
    backend=biber,
    style=ieee,
    citestyle=numeric-comp,
    sorting=none,
    giveninits=true,
    natbib,
    hyperref,
    maxbibnames=99,
    doi=false,isbn=false,url=false,eprint=false
]{biblatex}

\usepackage{amsthm}
\usepackage[normalem]{ulem}
\usepackage{soul}

\newtheorem{theorem}{Theorem} 
\newtheorem*{temp}{Risk control with uniform bounds} 
\newtheorem{prop}{Proposition}
\newtheorem{lemma}{Lemma}
\newtheorem{cor}{Corollary}[section]

\usepackage{chngcntr}
\usepackage{apptools}
\AtAppendix{\counterwithin{theorem}{section}}
\AtAppendix{\counterwithin{prop}{section}}

\usepackage{mathtools}

\usepackage[dvipsnames]{xcolor}

\usepackage{tikz}
\usetikzlibrary{shapes,arrows, positioning}

\newcommand{\RR}{\mathbb{R}}
\newcommand{\EE}{\mathbb{E}}

\newcommand{\Dcoco}{\mathcal{D}_\textup{COCO}}

\newcommand{\argmin}{\textup{argmin}}

\newcommand{\cY}{\mathcal{Y}}
\newcommand{\cX}{\mathcal{X}}

\newcommand{\cB}{\mathcal{B}}
\newcommand{\cS}{\mathcal{S}}
\newcommand{\cH}{\mathcal{H}}

\newcommand{\cC}{\mathcal{C}}
\newcommand{\cD}{\mathcal{D}}
\newcommand{\cF}{\mathcal{F}}
\newcommand{\cT}{\mathcal{T}}
\newcommand{\cG}{\mathcal{G}}

\newcommand{\cN}{\mathcal{N}}

\newcommand{\PP}{\mathbb{P}}

\newcommand{\cZ}{\mathcal{Z}}

\newcommand{\Lhatp}{{\hat L}^+}

\DeclareMathOperator{\argmax}{argmax}
\newcommand{\defn}{:=}
\newcommand{\twomax}[2]{#1 \vee #2}
\newcommand{\e}{\mathrm{e}}
\newcommand{\bG}{\mathbb{G}}

\makeatletter
\def\blfootnote{\xdef\@thefnmark{}\@footnotetext}
\makeatother

\renewcommand*{\thefootnote}{\fnsymbol{footnote}}

\newcommand{\Btglob}{\cB_2(t, [0,1])}
\newcommand{\Btloc}[1]{\cB_1(t, #1)}
\newcommand{\predThat}{\widehat{\cT}_r^+}
\newcommand{\predT}{{\cT}_r^+}
\newcommand{\That}{\widehat{\cT}_r}
\newcommand{\trueT}{\cT_r}
\newcommand{\qglob}{q_\mathrm{glob}}
\newcommand{\hatqglob}{\hat q_\mathrm{glob}}
\newcommand{\qloc}[1]{q_\mathrm{loc}(#1)}
\newcommand{\hatqloc}[1]{\hat q_\mathrm{loc}(#1)}

\title{Data-Adaptive Tradeoffs among Multiple Risks in Distribution-Free Prediction}
\date{\vspace{-0.2cm} March 2024 \vspace{-0.5cm}}

\author[1]{Drew T. Nguyen$^{*}$}
\author[2]{Reese Pathak$^{*}$}
\author[2]{Anastasios N. Angelopoulos}
\author[3]{Stephen Bates}
\author[1]{Michael I. Jordan}
\affil[1]{Department of Statistics, UC Berkeley}
\affil[2]{Department of EECS, UC Berkeley}
\affil[3]{Department of EECS, MIT}
\begin{document}
\maketitle
\def\thefootnote{*}\footnotetext{Equal contribution. Contact: \texttt{drew.t.nguyen@berkeley.edu}, \texttt{pathakr@berkeley.edu}}\def\thefootnote{\arabic{footnote}}
\begin{abstract}
        Decision-making pipelines are generally characterized by tradeoffs among various risk functions. It is often desirable to manage such tradeoffs in a data-adaptive manner. As we demonstrate, if this is done naively, state-of-the art uncertainty quantification methods can lead to significant violations of putative risk guarantees. 
    To address this issue, we develop methods that permit valid control of risk when threshold and tradeoff parameters are chosen adaptively. Our methodology supports monotone and nearly-monotone risks, but otherwise makes no distributional assumptions. 
    To illustrate the benefits of our approach, we carry out numerical experiments on synthetic data and the large-scale vision dataset MS-COCO.   
\end{abstract}

\section{Introduction}\label{sec:intro}

In modern machine learning, a focus on complex prediction models and autonomous decision-making is typical, reflecting the engineering focus of its practitioners. However, guaranteeing that the quality of these predictions (or decisions) are within desired tolerances requires good uncertainty quantification (UQ), a classically statistical issue.  
A burgeoning literature on conformal prediction proposes a solution for this problem, based on treating these complex predictors as unknown black boxes~\cite{angelopoulosGentleIntroductionConformal2022a}. 

Many such ``black box'' conformal methods, as applied in supervised learning, are able to use an auxiliary \emph{calibration dataset} $\{(X_i, Y_i)\}_{i=1}^n \subset \cX \times \cY$ to produce a \emph{prediction set} which is a subset of the label space $\cY$, for a guarantee of the form 
\[
\mathbb{P}\Big[R(\cC(X_{n+1}), Y_{n+1}) \leq \alpha\Big] 
\geq 1-\delta,
\]
where the risk function $R \colon 2^\cY \times \cY \to \mathbb{R}$ measures the quality of a prediction 
set in containing the true label. For example, in a $K$-class, multi-label classification setting, the prediction set $\cC(X)$ could correspond to possible 
classes for input $X$, and the risk 
may be a \emph{false negative rate}: the proportion of positive classes in the labels $Y$ that are missed by the elements of the prediction set $\cC(X)$. In practice, these predictions often constitute the final decision made by an autonomous system, and an appropriate risk measures the consequences of incorrect decisions. 

In this work, we attempt to address a major shortcoming of  existing UQ methods as described above:
they typically assume that the data analyst has selected the tolerance level $\alpha$ in advance of observing any data, and has already determined a risk $R$ to control. Practitioners, on the other hand, 
often select tolerance levels in a data-dependent way, invalidating the guarantees of UQ methods and hindering their applicability. This holds especially for complex machine learning systems, which are expensive to train and tune. 

\begin{figure}
    \centering
    \includegraphics[width = \textwidth]{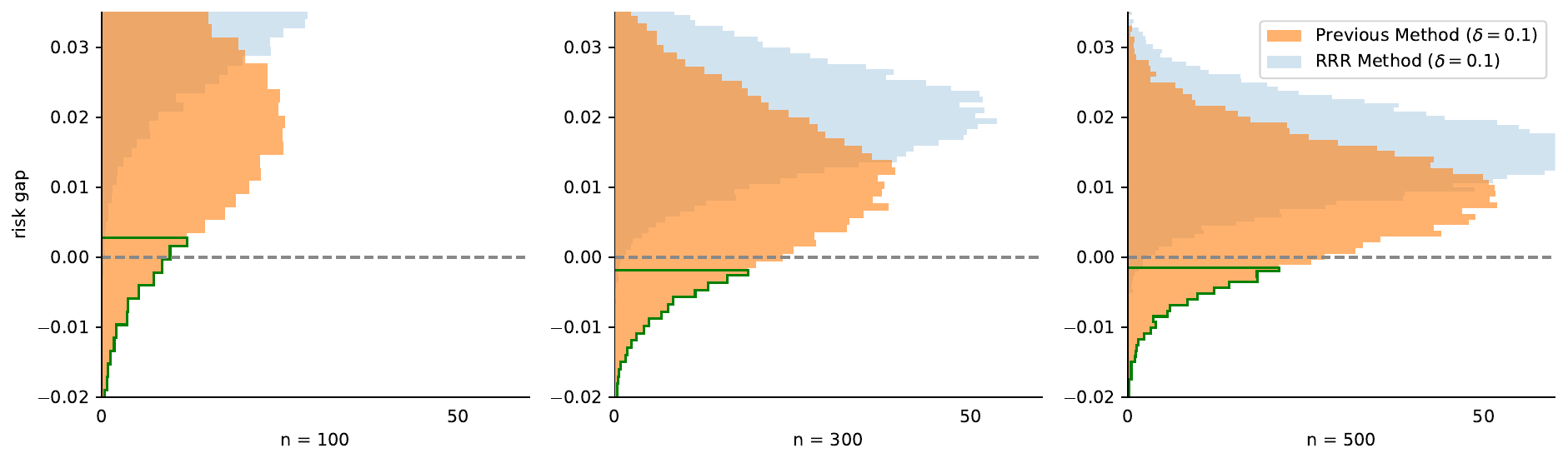}
    \caption{Histograms for 20K realizations of the risk gap $\alpha - \texttt{FNR}$ under a data-dependent choice of $\alpha$, given observations of $n$ calibration points. Each histogram represents different ways to set $\alpha$. The solid outline has total area $\delta = 0.1$.} \label{fig:mscoco_violation}
\end{figure}

Despite this, one may naïvely hope that the guarantees hold approximately well, so as to be useful in an engineering setting. That is, a system builder may wish to proceed without considering the data-dependence in the choice of a tolerance level $\alpha$, hoping that this will not lead to undesirable violations in practice. But is this actually true? In general, no. 
In Figure~\ref{fig:mscoco_violation}, we present the results of an experimental case study, discussed further in Section~\ref{sec:violation}, in which we measure the risk gap---the difference between $\alpha$ and the false negative rate (FNR) of a deployed multi-label classifier---when applying methods from the papers~\cite{batesDistributionfreeRiskcontrollingPrediction2021a, angelopoulosLearnThenTest2022b} with a data-dependent choice of $\alpha$ and with the exceedance probability parameter $\delta$  set to $0.1$. 
The histogram corresponding to this method shows the FNR exceeding $\alpha$ with probability greater than $\delta = 0.1$, which runs counter to the conservative finite-sample guarantees of \cite{batesDistributionfreeRiskcontrollingPrediction2021a} that hold for fixed $\alpha$. In fact, the exceedance occurs with empirical probability $0.14$ when $n = 300$ and $n = 500$. The high-level problem here is clear: selecting $\alpha$ after the observation of data can lead to a notable lack of control.

The present work addresses this problem by providing risk guarantees that account for data-dependent, post hoc choices of $\alpha$.
Figure~\ref{fig:mscoco_violation} also depicts the risk gap for a method we propose called \emph{restricted risk resampling} (RRR), which controls the FNR even after $\alpha$ is chosen to optimize a tradeoff. Though the risk gap may seem large,\footnote{When $n = 500$, the risk gap of RRR is larger than that of the previous method by about $0.01$ on average.} this reflects the flexibility of RRR---it allows the analyst to revise their choice of $\alpha$ \emph{in any way based on the calibration data} while still retaining a guarantee of low risk. This is particularly useful when data is scarce, as in medical diagnostics, and a classifier is employed at multiple different sensitivity levels. The underlying mathematical tool here is a simultaneity guarantee, and one of our main contributions is to develop  methods that possess a general simultaneity property, by establishing new theoretical results regarding the uniform convergence of monotone functions.

Section \ref{sec:method} is the core of this work. It states the key results in the form of uniform confidence bounds, including a functional analog of an inequality of~\cite{bentkusHoeffdingInequalities2004}. We adapt these results, as corollaries, into several risk-control procedures, valid for monotone losses and risks: (1) control via a nonasymptotic upper bound, (2) \textit{risk resampling}, a bootstrap-based procedure which is asymptotically exact, (3) \textit{restricted risk resampling}, a refinement of risk resampling which optionally ignores large choices of $\alpha$, and (4) extensions to certain classes of non-monotone functions.

The remainder of the paper supports these core results. Section \ref{sec:experiments} contains experiments using the results of Section \ref{sec:method}. The theoretical underpinnings of Section \ref{sec:method}, regarding empirical process theory for monotonically-indexed function classes, are presented in \Cref{sec:bounds}, with careful proofs deferred to Appendix \ref{app:proofs}. We discuss and conclude in Section~\ref{sec:discussion}.

\subsection{Case study: risk tradeoffs and risk control on MS-COCO}\label{sec:violation}

\begin{figure}
    \centering
    \includegraphics[width = \textwidth]{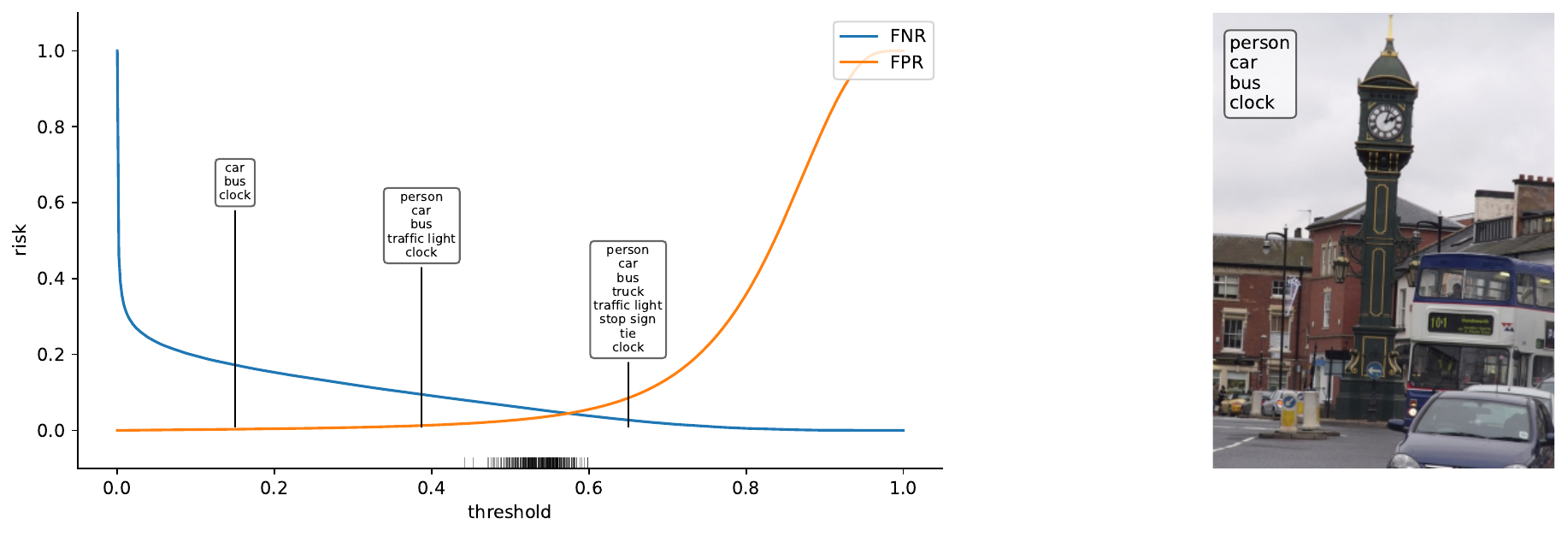}
    \caption{Prediction sets on an example from MS-COCO using the classifier $\cC_t$ of Section \ref{sec:violation}, as the threshold $t$ is varied. Also plotted is the FNR and FPR averaged over 60K held-out images, which are unobserved in practice; additionally, a rug plot of the scheme $t = \hat t$ optimized as in equation \eqref{eq:eventradeoff}, based on simulation draws of size $n = 500$.} \label{fig:mscoco_risks}
\end{figure}

We provide a basic description of the experimental setup presented in Figure~\ref{fig:mscoco_violation}, with complete details deferred to Appendix \ref{app:details_violation}.  The 2014 MS COCO dataset \cite{linMicrosoftCOCOCommon2015} consists of images $X \in \cX$ depicting everyday scenes in which any number of $K = 80$ common objects may be present (e.g., \texttt{dog}, \texttt{train}, \texttt{chair}). The label of each image $X \in \cX$ is a vector $Y \in \{0,1\}^K$, corresponding to the 80 classes which may be present in the image. The task of predicting such a $Y$ is called \textit{multi-label} classification. We now outline how we used this dataset in the experiment presented in Figure \ref{fig:mscoco_violation}.

To create a predictor, we trained a neural network that outputs scores, $f(X) \in [0,1]^K$, such that a higher score $[f(X)]_k$ on the $k$th component roughly corresponds to a higher chance that $Y_k = 1$. The final classification is then performed with a threshold classifier, implemented as $\cC_t : \cX \to \{0,1\}^K$ with $k$-th component
\begin{equation}\label{eq:threshold-classifier}
    [\cC_t(X)]_k = 1\{[f(X)]_k > (1 - t)\}.
\end{equation}
For a classifier $\cC_t$, and a labeled image $(X, Y)$, define the false negative proportion as $\ell(t; (X,Y)) = \#\{k : [\cC_t(X)]_k = 0\} / \#\{k: [Y]_k = 1\}$, which is the number of false negatives over true positives, using the threshold $t$. The false positive proportion $q(t, (X,Y)) = \#\{k : [\cC_t(X)]_k = 1\} / \#\{k: [Y]_k = 0\}$ is likewise the number of false positives over true negatives, using the threshold $t$. 

After training, we observe $n$ additional calibration data points $\{(X_i, Y_i)\}_{i = 1}^n$. In this example, we wish to use them to choose the threshold $t$ to \textit{trade off} the false negative rate (FNR) and the false positive rate (FPR), denoted as $L(t)$ and $Q(t)$, respectively:
\[
    L(t) = \EE[\ell(t, (X,Y))], \quad Q(t) = \EE[q(t, (X,Y))].
\]
In real problems, any scheme used to choose the threshold tends to be based on data visualization and intuitive consideration of the problem domain. For concreteness we make a stylized choice here, one that could be plausibly implemented by a practitioner.  Specifically, given the $n$ observations, we take $\hat t$ to trade off the empirical risks evenly:
\begin{equation}\label{eq:eventradeoff}
    \hat t = \underset{{t \in [0,1]}}{\argmin}\sum_{i = 1}^n \ell(t; (X_i,Y_i)) + q(t; (X_i, Y_i)).
\end{equation}
As context for this choice of $\hat t$, see Figure \ref{fig:mscoco_risks} for an illustration of the tradeoffs inherent in this problem.

We would also like to use the same $n$ calibration data points to \textit{control} these risks. For example, we say the FNR is controlled below $\alpha$ with probability at least $1 - \delta$ if 
\begin{equation}\label{eq:controldef}
    \PP\Big(L(\hat t) \leq \alpha \Big) \geq 1 - \delta.
\end{equation}
In previous work \cite{batesDistributionfreeRiskcontrollingPrediction2021a} such control is achieved via the following approach: choose the threshold with a special scheme $\hat t$ such that the upper bound of \cite{waudby-smithEstimatingMeansBounded2023} is below $\alpha$ for all $t \geq \hat t$.
In general, however, such an approach is inadequate, because the control level $\alpha$ can be data-dependent or random.  In particular, in the current example our choice of $\hat t$ is determined by an estimated risk-reward tradeoff, as is typical in applications.

We now need to plug in some quantity for $\alpha$ satisfying equation $\eqref{eq:controldef}$, which simply amounts to using the tightest upper bound available for the random variable $L(\hat t)$. Here we will consider two alternatives. First, we can ignore the data dependence in $\alpha$ and use the tight bound of \cite{waudby-smithEstimatingMeansBounded2023}, which is valid for fixed $\alpha$. Call this value $\alpha_1$. Second, we can apply the bound that comes from the RRR procedure of Section \ref{sec:loc-sim}, specifically by plugging the chosen threshold $\hat t$ into the right-hand side of Equation \eqref{eq:locsim-res1}---call this $\alpha_2$. 

Figure 1 plots the histograms for the random variables $\alpha_1 - L(\hat t)$ in front and $\alpha_2 - L(\hat t)$ behind. The lack of control shown by the histogram in front, and the difficulty of formalizing the choice of tradeoffs in practice, motivates our proposal of the RRR method, which is valid even without an explicit scheme such as Equation \eqref{eq:eventradeoff}.

\subsection{Prior work}

Our work belongs to the general body of work in distribution-free, frequentist  uncertainty quantification, as exemplified by conformal prediction \cite{vovkAlgorithmicLearningRandom2022c,angelopoulosGentleIntroductionConformal2022a}. While our work is in the general vein of conformal prediction, it is also distinct in its focus on general risk functions, as in earlier work on risk-controlling prediction sets \cite{batesDistributionfreeRiskcontrollingPrediction2021a}. As in that work, we augment the classical framework of conformal prediction via multiple-testing-based arguments and concentration results. The multiple-testing aspect of our work draws on a long tradition of methodology for simultaneous inference in various problem domains, from Scheffé's method for inference on all contrasts in linear regression \cite{scheffe1953method} to more modern problems \cite{dickhausSimultaneousStatisticalInference2014,goemanMultipleTestingExploratory2011a,katsevichSimultaneousHighprobabilityBounds2020a,berkValidPostselectionInference2013}. We also draw from empirical process theory, where uniform versions of concentration results provide tools for simultaneous inference in general settings~\cite{shorackEmpiricalProcessesApplications2009,kosorokIntroductionEmpiricalProcesses2008a}.  

Our focus is hypothesis testing. Recall that in the standard Neyman-Pearson paradigm the methodological goal is to maximize power subject to type-I error control at pre-specified level $\alpha$.  It has been noted that this paradigm is unjustified from a decision-theoretic standpoint~\cite{lehmannUniformlyMostPowerful2005, tetenovEconomicTheoryStatistical2016,grunwaldNeymanPearson2023}. Instead, in the words of Lehmann and Romano, $\alpha$ should be chosen ``in relation to the attainable power'' \cite{lehmannUniformlyMostPowerful2005}.
Attempts to find principled ways to set $\alpha$ in hypothesis testing go back to the 1950s \cite{lehmannSignificanceLevelPower1958a}. There are also more direct connections between distribution-free inference and decision-making with more explicit utility maximization \cite{vovkConformalPredictiveDecision2018,straitouriImprovingExpertPredictions2023, lekeufackConformalDecisionTheory2023h}. 

There has been recent interest in the control of tradeoffs between multiple risks \cite{laufer-goldshteinEfficientlyControllingMultiple2022a,linFastOnlineValuemaximizing2023,teneggiHowTrustYour2023}, notably the trading off of set size and coverage in conformal prediction \cite{sadinleLeastAmbiguousSetValued2019a,dhillonExpectedSizeConformal2023}. This work generally falls within the Neyman-Pearson framework of constrained optimization and is distinct conceptually from our work.

Uniform bounds have been studied recently in distribution-free novelty detection problems~\cite{bates2023testing,gazin2023transductive}. The bounds obtained in that work focus on empirical cumulative distribution functions, which correspond to binary-valued losses in our context. We also note the work of \cite{sarkarPostselectionInferenceConformal2023a} who share our focus on tradeoffs in distribution-free uncertainty quantification via uniform bounds, but again, their scope is limited to binary losses. Further results on the binary setting can be found in the work of  \cite{vovkConditionalValidityInductive2012,park2020pac,bianTrainingconditionalCoverageDistributionfree2023,liang2023algorithmic, amann2023assumption}.

\section{Risk Control with Uniform Bounds}\label{sec:method}

In this section, we focus on uniform convergence results for monotone losses, and show how these results can be adapted for risk control. 

In Section~\ref{sec:setting-technical}, we set notation and set up our risk-control problem.
We then present two uniform convergence results,  the first of which (Section~\ref{sec:finite-sample-monotone}) is a nonasymptotic concentration inequality, and the second (Section~\ref{sec:asymptotic-monotone}) a functional central limit theorem. We then present corollaries of these results that demonstrate how they can be used to establish risk control in practice. The proofs of the main results are deferred to~\Cref{sec:bounds}, with technical arguments deferred to Appendix~\ref{app:proofs}.

These results are then extended to be more practically useful by relaxing the requirements of uniformity (Section \ref{sec:loc-sim}) and monotonicity (Section~\ref{sec:beyondmono}). 
Proofs for Section \ref{sec:loc-sim} are provided in Appendix \ref{app:zrnicproof}.

\subsection{Setting}\label{sec:setting-technical}
Consider an input space $\cX$, an output space $\cY'$, and a label space $\cY$.
Consider also a family of predictors, $\cC_t \colon \cX \to \cY'$, indexed by a real-valued parameter $t \in [0,1]$.
Furthermore, let $\ell \colon \cY' \times \cY \to [0,1]$ denote a bounded loss function.
We impose the condition that the loss is \emph{monotone} in the parameter $t$, meaning the following implication holds:
\begin{equation}
    \label{eq:monotonicity}
    \mbox{if}~t \leq s, \quad \mbox{then} \quad 
    \ell(\cC_t(x), y) \geq \ell(\cC_s(x), y), 
    \quad\mbox{for all}~x \in \cX, y \in \cY. 
\end{equation}

Finally, suppose we are given a dataset $\mathcal{D}_n := (X_1, Y_1), \dots, (X_n, Y_n)$, which is drawn i.i.d.\ from a probability distribution $P$. Before deployment of the predictor $\cC_t$, the user chooses $t = \hat t$ based on this dataset, for example by considering risk tradeoffs. Define the population and empirical risks as 
\[
L(t) := \EE_{(X, Y) \sim P}\big[\ell(\cC_t(X), Y)\big] \quad \mbox{and} 
\quad 
\hat L_n(t) = \frac{1}{n} \sum_{i = 1}^n \ell(\cC_t(X_i), Y_i),
\] 
respectively.  Our goal is to compute a bound $\alpha$ such that $\PP(L(\hat t) \leq \alpha) \geq 1 - \delta$, which means that the risk of $\cC_{\hat t}$ is controlled below $\alpha$ with probability at least $1 - \delta$.

A wide range of schemes may be used to choose $\hat t$ in practice, including schemes that are difficult to formalize, and thus we prefer not to target any specific scheme. Instead, our approach to risk control is based on upper confidence bounds that are \emph{uniform}.
\begin{temp}
For some range $\cT \subset [0,1]$, compute a upper confidence bound $\Lhatp_n(t)$ such that
\[
L(t) \leq  \Lhatp_n(t)
\quad \mbox{simultaneously for all}~t \in \cT, 
\]
with probability at least $1 - \delta$. Then for any $\hat t \in \cT$, the risk of $\cC_{\hat t}$ is controlled below $\alpha = \Lhatp_n(\hat t)$. 
\end{temp} 

\subsection{Finite-sample result for monotone losses}\label{sec:finite-sample-monotone}

We first present a finite-sample uniform concentration bound based on a novel extension of the Dvoretsky-Kiefer-Wolfowitz (DKW) inequality~\cite{massartTightConstantDvoretzkyKieferWolfowitz1990} to monotone losses. To state it, we recall some notation.
Associated with the population and empirical risks are the following rescaled and centered processes 
and associated suprema: 
\[
\bG_n(t) \defn \sqrt{n} (\hat L_n(t) - L(t)), \quad 
D_n^+ \defn \sup_{t} \bG_n(t), \quad \mbox{and} \quad 
D_n^- \defn \sup_t -\bG_n(t).
\]
We then have the following nonasymptotic concentration inequality.
\begin{theorem}\label{thm:nonasymp_risk}
For every $\lambda > 0$, we have 
\[
\twomax{\PP\{D_n^{-} > \lambda\}}{
\PP\{D_n^+ > \lambda\}}
\leq 
\e \, \exp(-2 \lambda^2). 
\]
\end{theorem}
\noindent See Section~\ref{sec:proof-thm-nonasymp} for a proof of this theorem. It can be viewed as a functional analog of the inequalities of \cite{bentkusHoeffdingInequalities2004}, in particular his equation (1.1). 

Note that the dependence on $\lambda$ is optimal, but the prefactor $\e$ appearing in Theorem~\ref{thm:nonasymp_risk} is known to be improvable (to 1) in the special case that $\ell$ is a indicator function, as shown by \cite{massartTightConstantDvoretzkyKieferWolfowitz1990}. 
Before proceeding, we record an immediate consequence of Theorem~\ref{thm:nonasymp_risk}: an upper confidence bound (which can be used for risk control), and a lower confidence bound. 
\begin{cor}[Nonasymptotic confidence bounds]\label{cor:finite-sample}
For every sample size $n \geq 1$ and any $\delta \in (0, 1)$ it 
holds that 
\begin{equation}\label{eq:bdkw-ctrl-UB}
L(t) \leq \hat L_n(t) + \sqrt{\frac{\log(\e/\delta)}{2 n}} 
\quad \mbox{simultaneously for all}~t, 
\end{equation}
with probability at least $1 - \delta$. Similarly, 
we have 
\begin{equation}\label{eq:bdkw-ctrl-LB}
L(t) \geq 
\hat L_n(t) - \sqrt{\frac{\log(\e/\delta)}{2 n}} 
\quad \mbox{simultaneously for all}~t,
\end{equation}
with probability at least $1 - \delta$. \noeqref{eq:bdkw-ctrl-LB}
\end{cor}

We note that this bound can be quite conservative for practical problems as it does not account for the variance in the process $\bG_n$, which can be substantially smaller than the worst case. In the extreme case of zero variance, meaning that the loss $\ell(C_t(X), Y)$ is a non-random function of $t$, then $D_n^{+} = D_n^{-} = 0$ deterministically, which Theorem~\ref{thm:nonasymp_risk} does not account for.

This issue also occurs for bounds based on standard concentration results such as Hoeffding's inequality. In practice, it is often more appealing to use confidence bounds based on an asymptotic normal approximation which do adapt to the variance, as we describe next. 

\subsection{Asymptotic result for monotone losses}
\label{sec:asymptotic-monotone}

Motivated by the lack of instance-adaptivity of the nonasymptotic concentration inequality, we now present a tighter uniform confidence band. It hinges on the following functional central limit theorem, proven in \Cref{sec:bounds}, using the monotonicity assumption on the losses.
\begin{theorem}\label{thm:fclt}
The rescaled, centered process $\bG_n$ converges in distribution to a centered Gaussian process $\bG$.
\end{theorem}
The next lemma is technical and is proven in Appendix \ref{app:lempropproof}. It essentially shows that to estimate the process $\bG$, we can sample with replacement from the data---in other words, a bootstrap approach suffices, asymptotically. To state it, we need to introduce some notation relating to the bootstrap.

We use the notation $\mathcal{D}_n^\star \defn \{(X_i^\star, Y_i^\star)\}_{i=1}^n$ to denote $n$ samples, drawn with 
replacement, from the original dataset $\mathcal{D}_n$. We define the \emph{bootstrap empirical risk} as follows:
\[
\hat L_n^\star(t) \defn \frac{1}{n} \sum_{i = 1}^n \ell(\cC_t(X^\star_i), Y^\star_i).
\]
and the associated rescaled and centered process as $\bG_n^\star(t) = \sqrt{n}\big(\hat L_n^\star(t) - \hat L_n(t))$. Denote $D_{n, \star}^{\pm} =\sup_t \pm \bG_n^\star(t)$ to be the supremum of 
this process, where the sign $\pm$ denotes either $+$ or $-$, and $D_{n, \star} = D_{n, \star}^+ \vee D_{n, \star}^-.$

\begin{lemma}\label{lem:bootvalid}
If $\bG_n$ converges in distribution to a Gaussian process $\bG$, then the conditional distribution of the process $\bG_n^\star \mid \cD_n$ converges to the distribution of $\bG$ in probability; also, the conditional distribution of the random variable $D_{n, \star}^{\pm} \mid \cD_n$ converges to the distribution of $\sup_t \pm \bG(t)$ in probability.
\end{lemma}

The following corollary is an immediate consequence of Theorem \ref{thm:fclt} and Lemma \ref{lem:bootvalid}, and the first result in the corollary, which is an upper confidence bound, can be used directly for risk control. We refer to the overall functional bootstrap procedure as \emph{risk resampling} (RR).

\begin{cor}[Confidence bounds via risk resampling]\label{cor:bootstrap_risk}

Fix $\delta \in (0, 1)$. If $\hat q$ satisfies $\PP(D_{n, \star}^{-} > \hat q \mid \mathcal{D}_n) \leq \delta$, as $n \to \infty$, we have with probability at least $1 - \delta$,  
\begin{equation}\label{eq:bootstrap-UB}
    L(t) \leq \hat L_n(t) + \frac{\hat q}{\sqrt{n}} \quad \mbox{simultaneously for all}~t.
\end{equation}
Alternately, if $\hat q$ satisfies $\PP(D_{n, \star}^{+} > \hat q \mid \mathcal{D}_n) \leq \delta$, then with probability at least $1 - (\delta + o(1))$, as $n \to \infty$, we have with probability at least $1 - \delta$,  
\begin{equation}
    L(t) \geq \hat L_n(t) - \frac{\hat q}{\sqrt{n}} \quad \mbox{simultaneously for all}~t,
\end{equation}
and if $\hat q$ satisfies $\PP(D_{n, \star} > \hat q \mid \mathcal{D}_n) \leq \delta$, then with probability at least $1 - (\delta + o(1))$, as $n \to \infty$, we have with probability at least $1 - \delta$,  
\begin{equation}
    |L(t) - \hat L_n(t)| \leq \frac{\hat q}{\sqrt{n}} \quad \mbox{simultaneously for all}~t.
\end{equation}

\end{cor}

When the quantity $\hat q$ denotes the conditional $1 - \delta$ quantile, $\inf_q \{q : \PP(D_{n, \star}^\pm > q \mid \cD_n) \leq \delta) \} $, it can be computed exactly, but only in principle; this is generally infeasible as it requires enumeration over all possible $\binom{2n - 1}{n}$ realizations of the bootstrap dataset $\mathcal{D}_n^\star$. In practice (and also in our experiments, as presented in Section~\ref{sec:experiments}) we suggest using Monte Carlo to approximate $\hat q$ via a sample quantile $\hat q_\textup{boot}$; note that this is just the usual, basic bootstrap procedure \cite{tibshiraniIntroductionBootstrap1994, davisonBootstrapMethodsTheir1997}. 

As the Monte Carlo error cannot be ignored, it is of interest to ask how many replicates $B$ are needed. At a minimum, $B$ should be chosen large enough that $\hat q_\textup{boot}$ is stable, conditional on $\cD_n$. A more principled rule of thumb could be to take $B$ large enough so that $|\hat q - \hat q_\textup{boot}| < 0.01 \hat q_\textup{boot}$ with high probability conditional on $\cD_n$, based on a DKW confidence band.\footnote{Specifically, let $F(x) = \PP(D_{n, \star}^\pm \leq x \mid \cD_n)$. Based on $B$ bootstrap replicates, the DKW inequality gives a confidence band $[C^-(x), C^+(x)]$, and if $\hat q^\pm = \inf_q \{q : 1 - C^\pm(q) \leq \delta) \}$, then $[\hat q^+, \hat q^-]$ contains both $\hat q_\textup{boot}$ and $\hat q$ with high probability; $B$ could be chosen until $\hat q^- - \hat q^+$ seems small.}

As shown in Section~\ref{sec:experiments}, we find that the resulting 
bootstrap quantile $\hat q$ is much smaller than the factor $\sqrt{\log(\e/\delta)}$ guaranteed by the nonasymptotic inequality; indeed, Lemma \ref{lem:bootvalid} says it is asymptotically the best possible. Another advantage to the bootstrap procedure from Lemma \ref{lem:bootvalid} is that it can automatically extend to a different choice of the index set of the parameter $t$, namely a subset $\cT \subset [0,1]$. We leverage this in the next section.

\subsection{Extension to localized, simultaneous risk control}\label{sec:loc-sim}

The bootstrap procedure of Section \ref{sec:asymptotic-monotone}, which mimics the data distribution using resamples from the empirical distribution, is flexible in that it allows for certain refinements. In this section, we illustrate one possible refinement which is inspired by~\cite{zrnicLocallySimultaneousInference2023}.

Recalling our setting as outlined in Section \ref{sec:setting-technical}, suppose that before deployment of a predictive algorithm $\cC_t$, the user only wishes to choose the parameter $t$ within a subset of $[0,1]$. For example, in Section \ref{sec:violation}, where $L(t)$ represents the FNR risk in image classification, the user may wish to restrict to the sublevel set
\[
\mathcal{T}_r \defn \Big\{t \in [0, 1] : L(t) \leq r \Big\}.
\]
Setting $r = 0.1$, say, would encode a belief that a good algorithm cannot have FNR greater than this.

Unsurprisingly, it is wasteful to use the uniform bounds of Section \ref{sec:finite-sample-monotone} and \ref{sec:asymptotic-monotone}, which account for coverage violations in all of $[0,1]$, and not just $\cT_r$ which may be significantly smaller. Note that the set $\cT_r$ is unknown, so in practice the user can only restrict themselves to a data-dependent set such as
\[
\widehat{\mathcal{T}}_r \defn \Big\{t \in [0, 1] : \hat L_n(t) \leq r \Big\}
\]
The approach in this section assumes the user has done this. We give bounds that are valid simultaneously for all $t \in \widehat{\mathcal{T}}_r$, rather than all $t \in [0,1]$. We find empirically that they much tighter than the previous bounds.

To develop this approach, we need to extend the notation of the previous section. Define the bootstrap suprema based on the rescaled and normalized bootstrap process 
$\bG_n^\star(t) = \sqrt{n}(\hat L_n^\star(t) - \hat L_n(t))$ on sets $\cT \subset [0, 1]$:
\[
D_{n, \star}^+(\cT) = 
\sup_{t \in \cT} \bG_n^\star(t), 
\quad 
D_{n, \star}^-(\cT) = 
\sup_{t \in \cT} -\bG_n^\star(t), 
\quad \mbox{and} \quad 
D_{n, \star}(\cT) = \twomax{D_{n, \star}^+(\cT)}{D_{n, \star}^{-}(\cT)}.
\]
Our approach begins by fixing a level $r$ and two tolerance parameters $\delta_{\rm glob}, \delta_{\rm loc} \in [0,1]$. We then proceed in three steps: 
\begin{enumerate}
    \item Global estimation: Select a $\delta_{\rm glob}$-bootstrap quantile $\hat q_{\rm glob}$ the two-sided supremum risk over $[0,1]$ satisfying
    \[
    \PP\Big\{D_{n, \star}([0, 1]) > \hat q_{\rm glob} \mid \mathcal{D}_n\Big\} \leq \delta_{\rm glob}. 
    \]
    \item Localization: Form an adjusted empirical sublevel set:
    \[
    \widehat{\mathcal{T}_{r}}^+ \defn 
    \Big\{ t \in [0, 1] :
    \hat L_n(t) \leq r + 2 \frac{\hat q_{\rm glob}}{\sqrt{n}}
    \Big\},
    \]
    containing the original empirical sublevel set $\widehat{\mathcal{T}_{r}}$.
    \item Local estimation: Select a $\delta_{\rm loc}$-bootstrap quantile $\hat q_{\rm loc}$ of the one-sided supremum risk over $\widehat{\mathcal{T}_{r}}^+$ satisfying
    \[
        \PP\Big\{D_{n, \star}^-(\widehat{\mathcal{T}_r}^+) > \hat q_{\rm loc} \mid \mathcal{D}_n\Big\} \leq \delta_{\rm loc},
    \]
    and use it to compute an upper confidence band.
\end{enumerate}

The method can be understood intuitively as follows. The first step computes by
how much the size of the set $\widehat{\mathcal{T}_r}$ should be increased, to 
obtain the corrected set $\widehat{\mathcal{T}}_r^+$ of the 
second step.\footnote{
    To appreciate why such a correction is necessary, consider a fixed set 
    $\cT \subset [0,1]$, such as $\cT = [0,1]$ from Section \ref{sec:asymptotic-monotone}. A confidence band 
    ${\cB(t; \cT) : [0,1] \to 2^{[0,1]}}$ which is \emph{uniformly valid} over $\cT$---in the sense that $L(t) \in \cB(t, \cT)$ for all $t \in \cT$, with high probability---is typically no longer uniformly valid 
    when a random set $\widehat{\cT}$
    is substituted for $\cT$.
}  
The third step essentially carries out the bootstrap quantile estimate from the previous section, but specializing to $\widehat{\mathcal{T}}_r^+$ for tighter bounds. Due
to the correction, the confidence set is valid over the original, smaller set $\widehat{\mathcal{T}_r}$.

The initial two-sided global estimation is key.  
To see why, 
suppose we have a 
confidence set that is valid for a fixed set $\cT$ when specializing over $\cT$. 
Then it is valid for any \emph{subset} of $\cT$ when specializing over any \emph{superset} of
$\cT$, even if these sets are data dependent. Since the two-sided estimation
quantifies how far 
$\hat L_n$ is---both above and below---from the unknown mean $L$, we might plausibly 
find the fixed set $\cT_r$ to be ``sandwiched'' as
$\widehat{\cT}_r \subset \cT_r \subset \widehat{\cT}_r^+$, so that our confidence
set can apply. 

The following theorem gives the precise form of a uniform upper confidence bound for risk control. We refer to its computation as \emph{restricted risk resampling} (RRR), but 
like risk resampling, note that it is a functional form of the bootstrap, though now combined 
with the localization idea of \cite{zrnicLocallySimultaneousInference2023}. 

\begin{theorem}[Confidence bound via restricted risk resampling]\label{thm:locsim_bootstrap_risk}
Let $r \in [0, 1]$. Fix confidence parameters $\delta_{\rm glob}, \delta_{\rm loc} \in [0,1]$ and set $\delta = \delta_{\rm glob} + \delta_{\rm loc}$. 
Then we have, with probability at least $1 -( \delta + o(1))$, 
\begin{equation}\label{eq:locsim-res1}
L(t) \leq \hat L_n(t) + \frac{\hat q_{\rm loc}}{\sqrt{n}} \quad \mbox{simultaneously for all}~t \in \widehat{\mathcal{T}}_{r}.
\end{equation}
\end{theorem}

\noindent See Appendix~\ref{app:zrnicproof} for a proof of this result.

Let us make a few remarks. First, regarding implementation: we use Monte Carlo approximation to compute the quantiles $\hat q_{\rm loc}, \hat q_{\rm glob}$, similarly to Section \ref{sec:asymptotic-monotone}, and we choose a ratio of $\delta_{\rm loc} / \delta_{\rm glob} = 9$ for the confidence parameters in our experiments.

Second, note that the theorem claims validity over the observed data-dependent
set $\That$, rather than the population set $\trueT$. The former is arguably more useful in applications, as $\trueT$ is not observed, but a guarantee involving the population set is possible with minor modifications. Specifically, if the procedure is run with the level $r' = r - {\hat q_{\rm glob}}/{\sqrt{n}}$, then in addition to the conclusion of Theorem \ref{thm:locsim_bootstrap_risk}, we additionally have the inclusion $\widehat{\mathcal{T}}_{r'} \subset \cT_r$ with high probability, asymptotically.

Finally, regarding motivation: our goal was to spend the error budget less wastefully, when the user prefers parameters $t$ such that $\hat L_n(t)$ is small. (By monotonicity, this means that $t$ is large). We note that in past work, this concern has been addressed differently, using bounds which are valid simultaneously for all $t$, but which have \emph{variable width}: tighter for large $t$, looser for small $t$. The bound in this section, in contrast, is \emph{fixed width}, but can still be tighter for large $t$, as it need not be valid when $t$ is small. 

One seemingly plausible approach to variable-width upper bounds in our setting is inspired by the Monte Carlo method of \cite{bates2023testing}: for some function $f(t; \gamma)$, corresponding to the width of the bound, which is decreasing in $t$ and increasing in $\gamma$, define $\Lhatp_n(t) := \hat L_n(t) + n^{-1/2} f(t; \hat \gamma)$, where $\hat \gamma$ satisfies $\PP(\forall t, \; \bG^\star_n(t) \leq \hat f(t; \hat \gamma) \mid \cD_n ) \geq 1 - \delta$. Observe that a fixed-width bound corresponds to $f(t, \gamma) = \gamma$.

We do not pursue this further in the present work, but note that the present approach selects a set $\That$,
rather than a function $f(t; \gamma)$. This is advantageous when the risk $L(t)$, and hence the set $\That$, is more interpretable than the parameter $t$.  

\subsection{Extension to non-monotone losses}\label{sec:beyondmono}

Our previous concentration inequalities and upper confidence bounds only apply to population risks arising as expectations of monotone losses. In this section, we briefly discuss two extensions that we can accommodate that involve losses which are
not monotone. 

\subsubsection{Combinations of monotone risks}
In many situations, the risk that we would like to control can be decomposed into multiple monotone components. Formally, suppose we are interested in 
controlling the composition of $k$ different risks,   
\[
L(t) \defn \Psi\Big(L_1(t), \dots, L_k(t)\Big),
\quad \mbox{where} \quad 
L_i(t) \defn \EE[\ell_i(\cC_t(X), Y)], 
\]
for $i = 1, \ldots, k$. We assume for simplicity that $\ell_i$ have range
in $[0, 1]$ and are monotone in the sense of display~\eqref{eq:monotonicity}, 
but the overall function $\Psi \colon [0, 1]^k \to [0, 1]$ may possibly be non-monotone.

\paragraph{General approach:} To obtain a upper confidence bound on $L(t)$, we simply combine the uniform lower and upper confidence bounds for each of the components $L_i$, in the 
following two steps.
\begin{enumerate} 
\item Develop a $1 - \delta_{n, i}$ confidence band $\widehat{\mathcal{C}_i}$ for each $i$, 
such that 
\[
L_i(t) \in \widehat{\mathcal{C}_i}(t) \quad \mbox{simultaneously for all}~t \in \mathcal{T}, 
\]
holds with probability at least $1 - \delta_{n, i}$. 
\item 
Aggregate the confidence parameters and confidence sets by defining
\[
\delta = \sum_{i=1}^k \delta_{n, i}, 
\quad 
\widehat{C}_{\rm low}(t) \defn 
\inf_{\ell_i \in \widehat{\mathcal{C}_i}(t)} 
\Psi(\ell_1, \dots, \ell_k),
\quad \mbox{and} \quad 
\widehat{C}_{\rm up}(t) \defn \sup_{\ell_i \in \widehat{\mathcal{C}_i}(t)} 
\Psi(\ell_1, \dots, \ell_k).
\] 
\end{enumerate} 
Clearly, we have with probability at least $1 - \delta$ that 
\[
\widehat{C}_{\rm low}(t) 
\leq L(t) \leq 
\widehat{C}_{\rm up}(t), \quad \mbox{simultaneously for all}~t \in \mathcal{T}.
\]
For the first step in the above approach, we can apply any of our previously described confidence bounds since the component risks $\{L_i\}$ are bounded and monotone. 

\paragraph{Illustration for selective classification:} 
Consider the case where $L(t) = L_1(t)/L_2(t)$, which is a special case of the above
approach, having taken $L = \Psi(L_1, L_2)$ and $\Psi(\ell_1, \ell_2) = \ell_1/\ell_2$. The two risks can be define to capture a tradeoff, or may arise directly from the specification of $L$. 

This setup arises specifically in \emph{selective classification}, also known as 
classification with abstention, or classification with a reject option~\cite{chow1957optimum,chow1970optimum,cortes2016learning}. 
In this setting, we wish to classify covariates $x$ as falling into one of $K$ classes or, alternatively, we can \emph{abstain}, for instance, if we believe we are too uncertain to commit to a point prediction.  In this case, let $\cY' = \{1, \dots, K\} \cup \texttt{abstain}$. A classifier $\cC_t: \cX \to \cY'$ can work on top of learned scores $\hat p(X) \in \Delta^K$ in the $K$-simplex as follows. Denoting $k^*(x) = \argmax_k \hat p_k(x)$, it returns the highest scoring class unless a score threshold $t$ is not reached:
\[
    \cC_t(X) = \begin{cases}
                     k^*(X) & \hat p_{k^*(X)}(X) > t \\
                    \texttt{abstain} & \text{otherwise}.
                \end{cases}
\]
Then the following risk $L$, which is the probability of misclassification given that $\cC_t$ did not abstain, can be upper bounded with high probability by upper bounding the numerator and lower bounding the denominator:
\[
    L(t) = \PP[\cC_t(X) \neq Y \mid \cC_t(X) \neq \texttt{abstain}] = \frac{\PP[\cC_t(X) \neq Y, \cC_t(X) \neq \texttt{abstain}] }{\PP[\cC_t(X) \neq \texttt{abstain}]}.
\]

\subsubsection{Nearly monotone risks} 
Now we consider the case where we have a risk $L(t) = \EE[\ell(\cC_t(X), Y)]$ which is the expectation of a non-monotone loss $\ell$. 
The following approach will allow us to provide meaningful risk control when $L$ is ``nearly'' monotone. 

\paragraph{Monotonizing the loss:} 
Our approach is to \emph{monotonize the loss}. 
Formally, we define the functions
\[
\ell^\downarrow(\cC_t(X), Y) 
\defn 
\inf_{s \leq t} \ell(\cC_s(X), Y) 
\quad \mbox{and} \quad 
\ell^\uparrow(\cC_t(X), Y) 
\defn 
\sup_{s \leq t} \ell(\cC_s(X), Y). 
\] 
The functions $\ell^\downarrow$, $\ell^\uparrow$ are essentially the running minimum and maximum, respectively, over the set $\{s \leq t\}$. 
We define 
\[
L^\downarrow(t) \defn 
\EE[\ell^\downarrow(\cC_t(X), Y)] 
\quad \mbox{and} \quad 
L^\uparrow(t) \defn 
\EE[\ell^\uparrow(\cC_t(X), Y)].
\]
Since $\ell^\downarrow \leq \ell \leq 
\ell^\uparrow$, we also have 
$L^\downarrow \leq L \leq L^\uparrow$.

If the loss $\ell$ is bounded, then the functions $L^\downarrow, L^\uparrow$ satisfy the monotonicity assumptions needed to develop our lower and upper confidence bounds. In particular, we define the monotonized empirical risks,
\[
\hat L_n^\downarrow(t) \defn 
\frac{1}{n} \sum_{i=1}^n 
\ell^\downarrow(\cC_t(X), Y) \quad 
\mbox{and} \quad 
\hat L_n^\uparrow(t) \defn 
\frac{1}{n} \sum_{i=1}^n 
\ell^\uparrow(\cC_t(X), Y). 
\]
Then, using $\hat L_n^\downarrow, \hat L_n^\uparrow$, we can develop, respectively, simultaneous lower and upper confidence bounds on the risk $L$ using the inequalities developed in the previous sections. 

\paragraph{Batch-and-monotonize:} 
One concern that we may have with the approach developed above is that even when $L$ is close to monotone, the loss $\ell$ may be far from monotone, resulting in a larger than desired  gap between the population risks $L$ and the monotonized variants $L^\downarrow, L^\uparrow$. A way to address this is to \emph{batch} the samples, and then monotonize on these individual batches.

Specifically, using $k$ data points at a time can get tighter bounds. Assuming for simplicity that $n$ is divisible by $k$, we define for $j \in \{0, 1, \dots, n/k-1\}$ the dataset and loss
\[
Z_j = \{(X_{kj + i},Y_{kj + i})\}_{i=1}^k 
\quad \mbox{and} \quad 
\ell_k(t, Z_j) =  \frac{1}{k} \sum_{i = 1}^{k} \ell(\cC_t(X_{kj + i}), Y_{kj + i}).
\] 
We can then monotonize $\ell_k$ as described in the previous paragraph, and use this to develop simultaneous lower and upper confidence bounds on the population risk. At first glance this may appear to be lossy because there are only $n/k$ data points $Z_j$; however, values of $\ell_k$ can be expected to have lower variance than $\ell$, so variance-aware methods such as Corollary \ref{cor:bootstrap_risk} and Theorem \ref{thm:locsim_bootstrap_risk} will adapt.

\section{Experiments and Examples}\label{sec:experiments}

In this section, we will demonstrate the performance of the upper bounds of Section \ref{sec:method}:
the nonasymptotic bound (NASM),\footnote{Corollary \ref{cor:finite-sample}, 
right-hand side of Equation \eqref{eq:bdkw-ctrl-UB}.}
risk resampling (RR),\footnote{Corollary \ref{cor:bootstrap_risk}, 
right-hand side of Equation \eqref{eq:bootstrap-UB}.} 
and restricted risk resampling (RRR).\footnote{Theorem \ref{thm:locsim_bootstrap_risk}, 
right-hand side of Equation \eqref{eq:locsim-res1}.}. Each method was run with confidence parameter $\delta = 0.1$; the RR and RRR methods are run with 1000 bootstrap resamples, and the RRR method was run with risk tolerance $r = 0.1$ and global and local parameters $\delta_\textup{glob} = 0.01, \delta_\textup{loc} = 0.09$.  

These methods each amount to different ways to compute a uniform upper bound, which we
denote generically as $\Lhatp(t)$. We investigate two settings: a fully simulated setting, 
as well as the MS COCO setting from the Introduction.  

Since these bounds are \textit{fixed-width} bounds of the form
$\Lhatp(t) = \hat L_n(t) + \hat q / \sqrt{n}$, they must satisfy
\[
    \PP\left(\sup_t L(t) - \hat L_n(t) \leq \frac{\hat q}{\sqrt{n}}\right) \geq 1 - \delta.
\]
Hence in each setting we provide a quantile plot of $\hat q / \sqrt{n}$ against the true $1 - \delta$ quantile of $D_n = \sup_t L(t) - \hat L_n(t)$, which constitutes the best possible fixed-width bound. 

We also display miscoverage metrics. Call the quantity $\PP(\exists t \text{ s.t. } L(t) > \Lhatp(t))$ the \emph{anywhere miscoverage probability}. Additionally, for the selected set $\hat S = \{t : \hat L_n(t) \leq 0.1\}$, call the quantity ${\PP(\exists t \in \hat S \text{ s.t. } L(t) > \Lhatp(t))}$ the \emph{selected set miscoverage probability}; we plot these two quantities in bar charts.

Though the resampling-based bounds are asymptotically exact, 
their use may seem unreasonable if they are very wide. 
Hence, on the MS COCO example, 
we model the behavior of an analyst who trades off two risks. Choosing the parameter $\hat t$ as a function of the data, call $\EE[\Lhatp(\hat t) - L( \hat t)]$ the \emph{average conservatism}; we also plot this in a bar chart. For details on the specific function $\hat t$, see Appendix \ref{app:lossrisk}.

In each setting, we consider multiple different loss functions, and
plot the results for each method applied on each loss. In
addition, we include a fourth method,\footnote{Theorem \ref{thm:wsr} 
Equation \eqref{eq:wsr-def}.}, referred to as ``pointwise,'' 
which is not uniformly valid, but only pointwise valid in the sense that 
$\PP(L(t) \leq \Lhatp(t)) \geq 1 - \delta$ for every $t$ and finite $n$. This method,
due to \cite{waudby-smithEstimatingMeansBounded2023}, gives remarkably tight 
estimates of means of bounded random variables, so we display it
as a benchmark against our methods which are uniformly valid.

Lastly, we note that these probabilities and expectations cannot be computed exactly on finite data, so in the MS COCO example we must compute surrogates based on splitting our datasets in halves into a holdout set $\cH$ and a sampling set $\cS$; refer to Section \ref{app:details_violation} for details. Additionally, the suprema and miscoverage quantities are computed for $t$ in a grid; for the Gaussian example, it is a grid on $[-3,3]$ with size $1000$, and for MS COCO it is a grid on $[0,1]$ of size $500$. 

Replication code can be found at \texttt{github.com/drewtnguyen/risk-tradeoffs-experiments}.

\subsection{Simulated data}

To define a completely synthetic monotone loss function, consider empirical CDFs on batches of data. Let $Z_1, \dots, Z_n$ be i.i.d., where each $Z_i$ is a batch of five equi-correlated Gaussians, having covariance with diagonal values $1$ and off-diagonals $\rho \in [-1, 1]$:
\[
    Z_i = (X_{i1}, \dots, X_{i5}) \sim \cN_5(0, \rho \mathbbm{1} + (1 - \rho) \mathbbm{I}).
\]
Define the loss $\ell(t, Z_i)$ as
\[
    \ell(t, Z_i) = \sum_{j = 1}^5 1\{X_{ij} \leq t\}
\]
Evidently the risk is just the standard normal CDF $L(t) = \Phi(t)$, even though varying $\rho$ constitutes different loss distributions (we choose $\rho \in \{-0.2, 0.2, 0.6\}$ in the experiments). Note that if $\rho = 0$, uniform upper bounds on $L$ could be obtained by standard arguments for CDFs. 

\begin{figure}
    \centering
    \includegraphics[width=\textwidth]{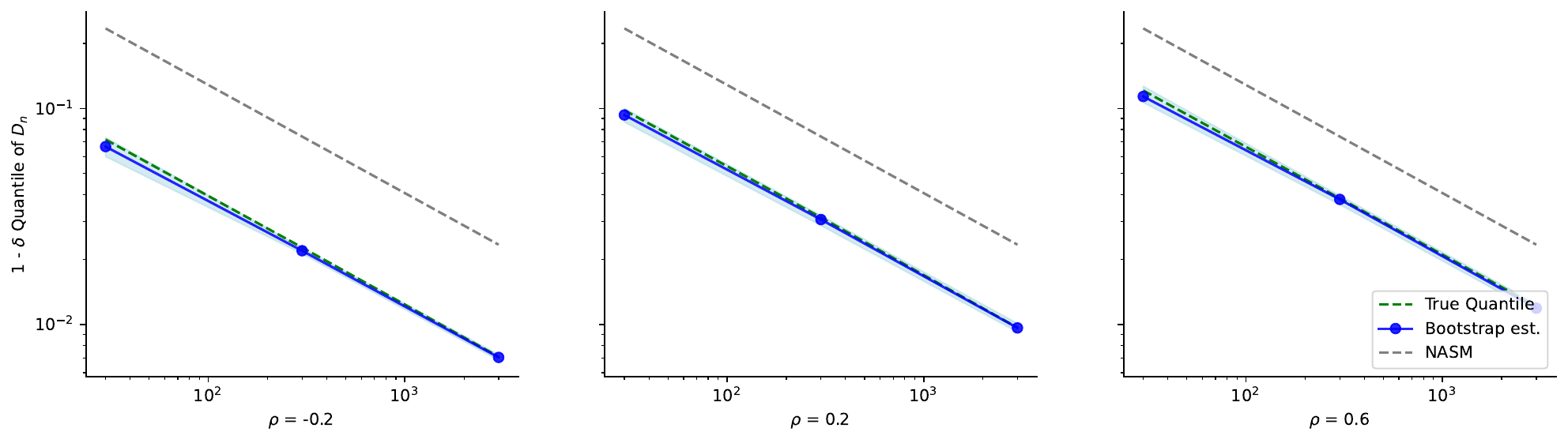}
    \caption{Simulated data: Log-log plot of true quantile of $D_n$ and its median bootstrap estimate (and 90/10\% quantiles) computed from 3K Monte Carlo runs.}\label{fig:gaussian_quantiles}
\end{figure}

\begin{figure}
    \centering
    \includegraphics[width=\textwidth]{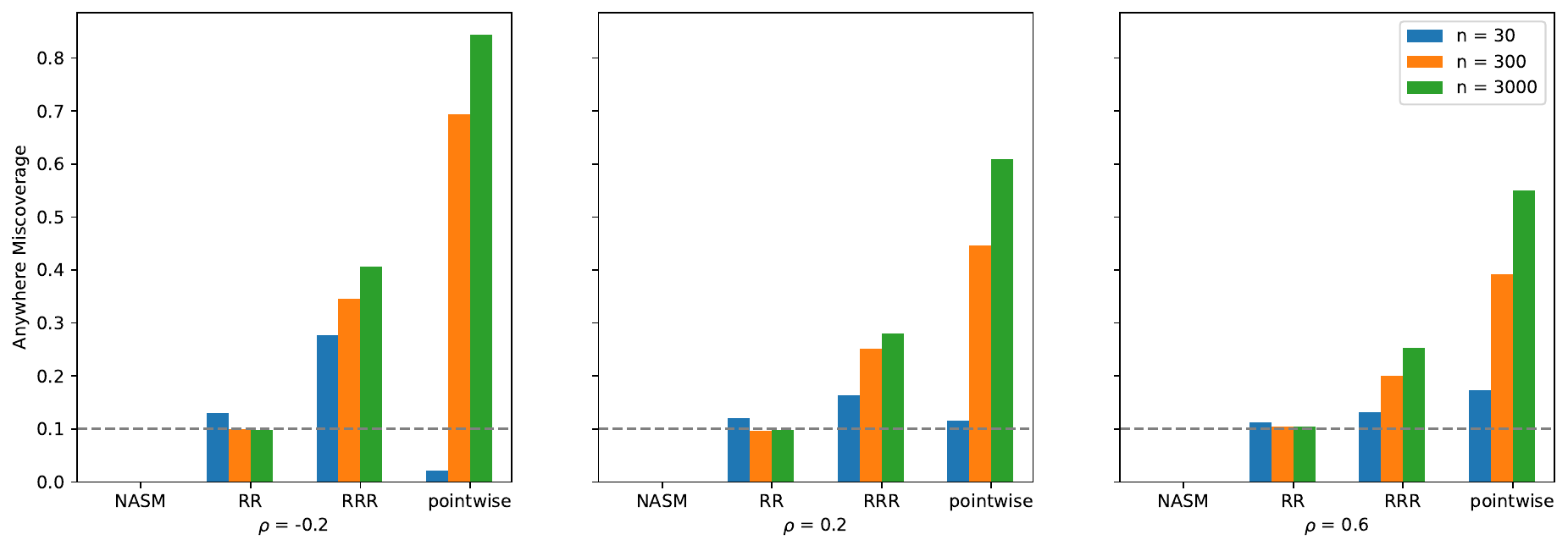}
    \caption{Simulated data: Anywhere miscoverage based on 20K Monte Carlo runs.}\label{fig:anywhere_gaussian}
\end{figure}

\begin{figure}
    \centering
    \includegraphics[width=\textwidth]{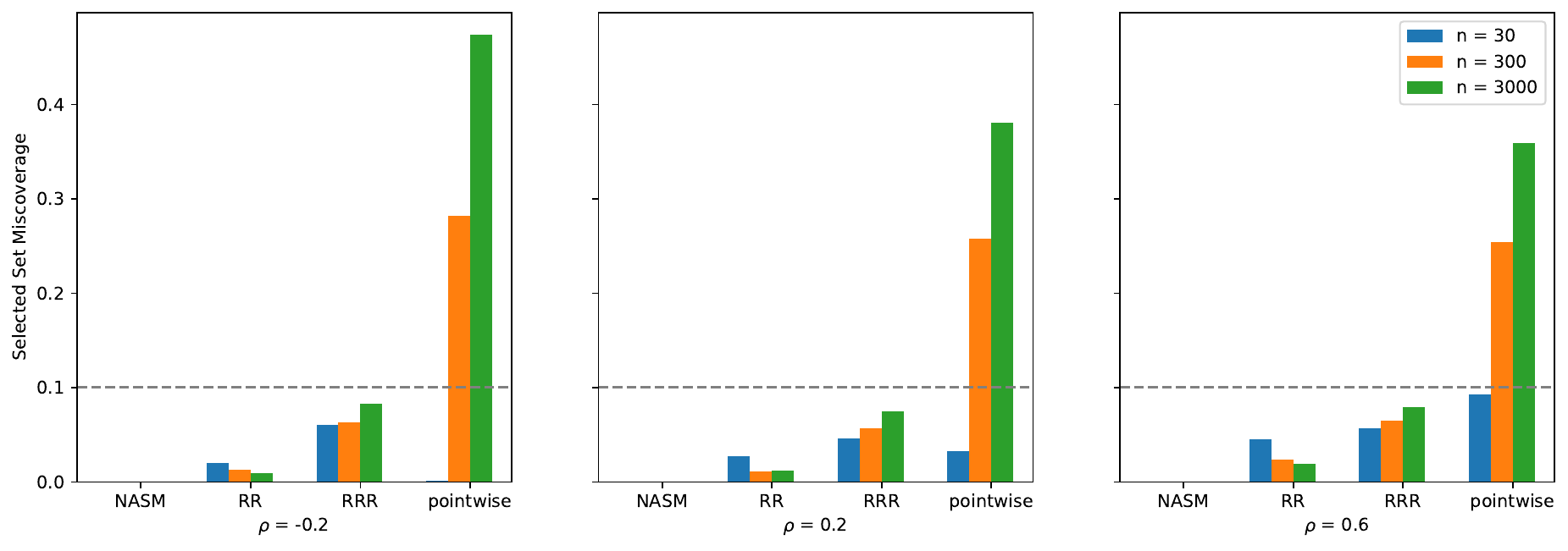}
    \caption{Simulated data: Selected set miscoverage based on 20K Monte Carlo runs.}\label{fig:selected_gaussian}
\end{figure}

The quantile plot of Figure~\ref{fig:gaussian_quantiles} shows the nonasymptotic upper bound, as well as convergence of the estimated bootstrap quantile (which can be predicted from Lemma \ref{lem:bootvalid}). The plots show that the nonasymptotic bound is, as predicted, valid for all sample sizes, but is very conservative; RR and RRR work for moderate sample sizes; and the pointwise bound is not 
uniformly valid. 

Note that the convergence of the bootstrap appears to be slower for small $\rho$, which is when the underlying process is closest to its mean. More generally, constant pre-factors in the convergence rate may depend on the problem setting. 

\subsection{MS COCO}

We revisit multi-label classification on the MS COCO data set that was discussed in the Introduction. The results are presented in Figures \ref{fig:mscoco_quantiles}-\ref{fig:conservatism_mscoco}, for four different multi-label classification risks of interest: FNR, FPR, FDR, and SetSize. 

The FNR and FPR were defined in the Introduction; the false discovery rate (FDR) 
is the expectation of the number of false positives over selected classes, 
while SetSize is the expectation of the normalized number of selected classes. 
For precise definitions and illustrations of these risks, see Appendix \ref{app:lossrisk}.

We can interpret control of the FDR as a guarantee that the selected classes are mostly true positives on average. Note that it is not a monotone risk, 
so we monotonize 
it as described in Section \ref{sec:beyondmono}. 

The figures are qualitatively similar to those of the previous section. Again, the nonasymptotic bound is extremely conservative and the pointwise baseline does not
have the right coverage, while resampling-based methods are valid at moderate sample sizes and are quite effective. In particular, Figure \ref{fig:conservatism_mscoco}, measuring the average tradeoff conservatism, shows that the RRR method does better than RR in terms of tightness of the bound, and compares favorably to the method that is pointwise valid.

\begin{figure}
    \centering
    \includegraphics[width=\textwidth]{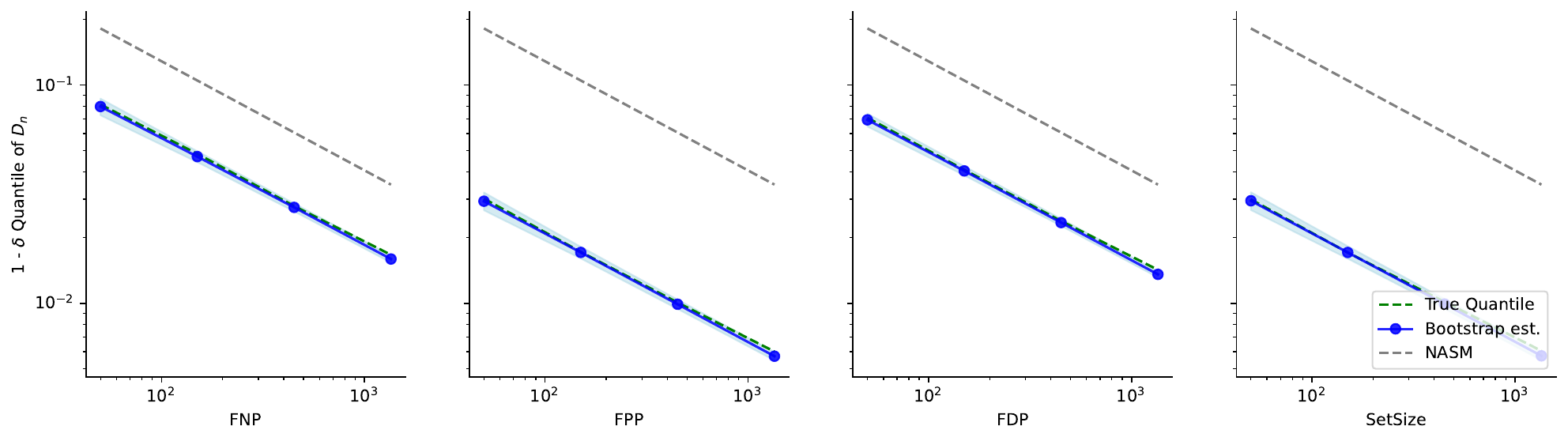}
    \caption{MS COCO data: Log-log plot of true quantile of $D_n$ and its median bootstrap estimate (and 90/10\% quantiles) computed from 3K Monte Carlo runs.}\label{fig:mscoco_quantiles}
\end{figure}

\begin{figure}
    \centering
    \includegraphics[width=\textwidth]{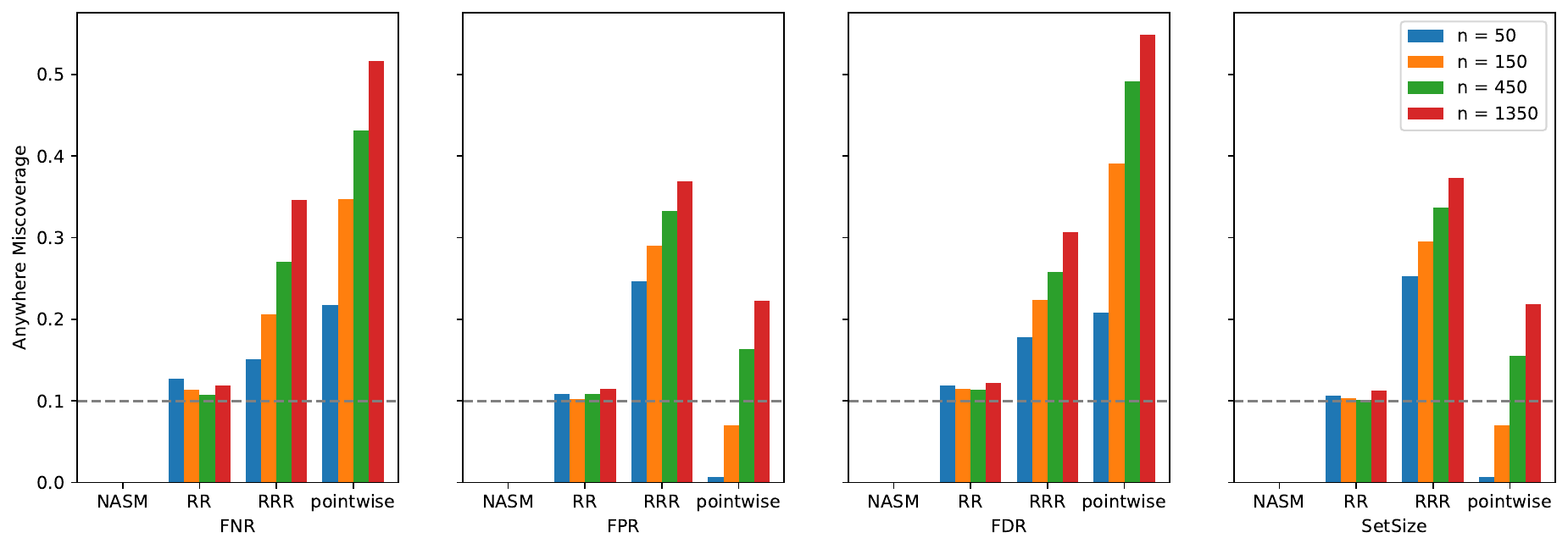}
    \caption{MS COCO: Anywhere miscoverage based on 20K Monte Carlo runs. (Miscoverage for larger $n$ is an artifact of having a finite population---see Appendix \ref{app:details_violation}).}\label{fig:anywhere_mscoco}
\end{figure}

\begin{figure}
    \centering
    \includegraphics[width=\textwidth]{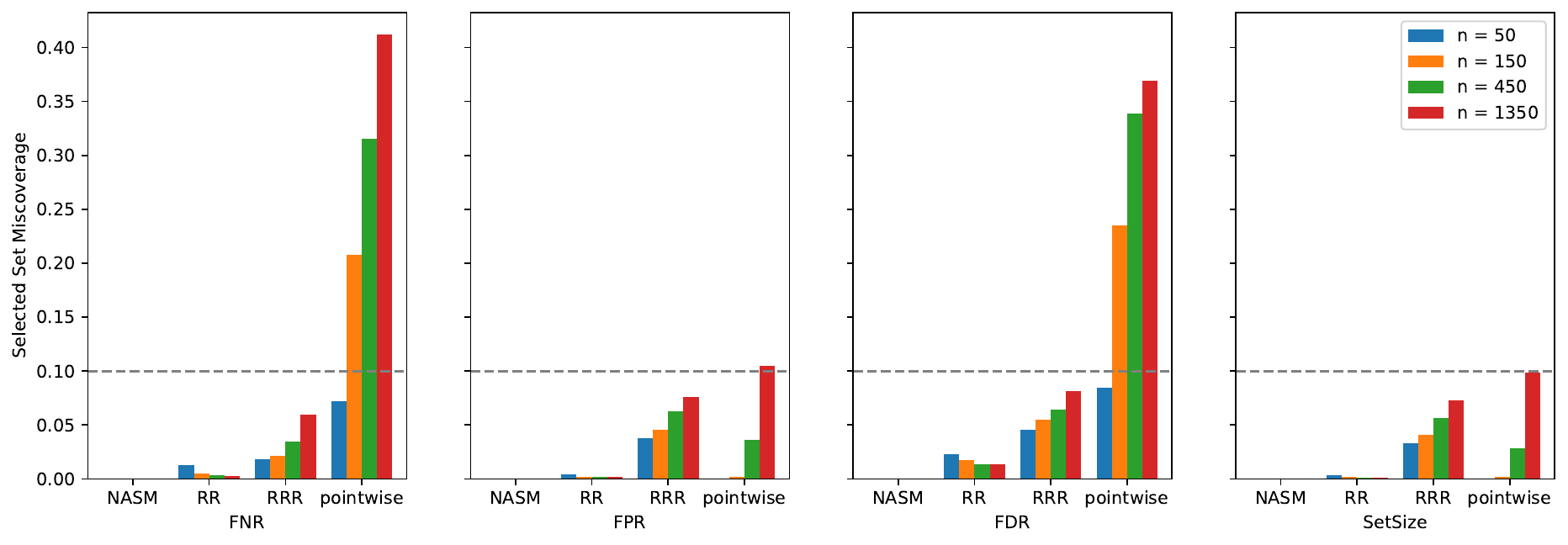}
    \caption{MS COCO: Selected set miscoverage based on 20K Monte Carlo runs.}\label{fig:selected_mscoco}
\end{figure}

\begin{figure}
    \centering
    \includegraphics[width=\textwidth]{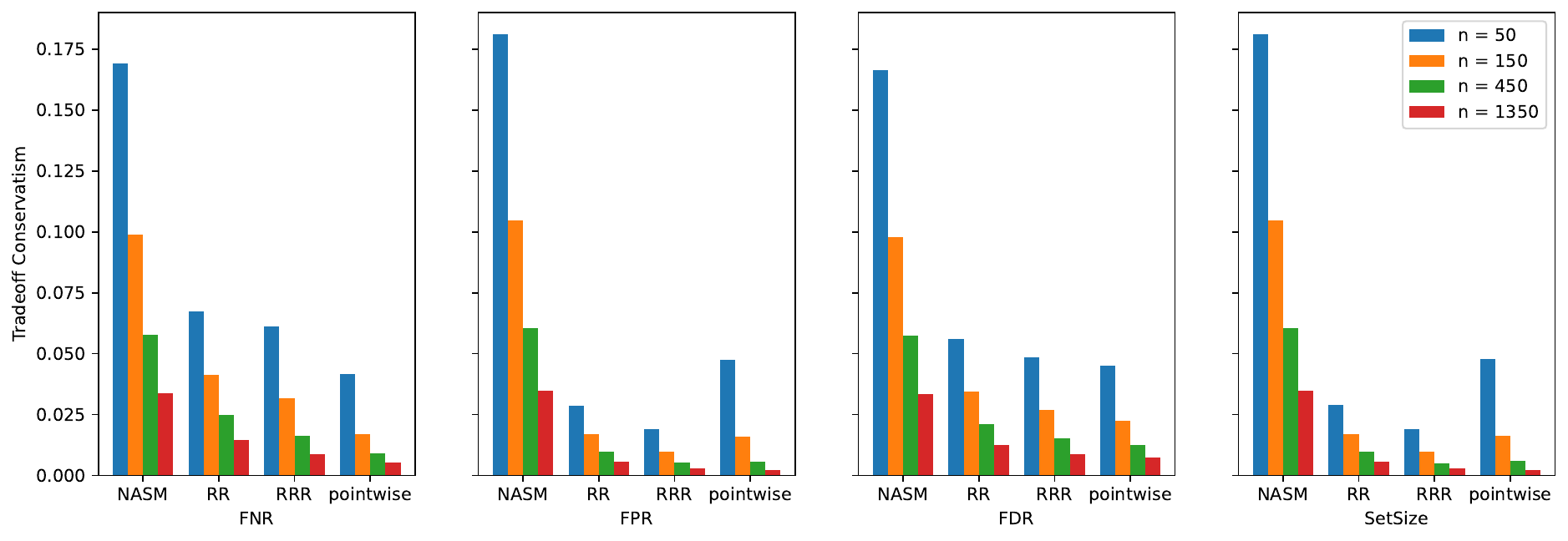}
    \caption{MS COCO: Tradeoff conservatism based on 20K Monte Carlo runs.}\label{fig:conservatism_mscoco}
\end{figure}

\section{Theoretical Underpinnings}\label{sec:bounds}

This section is a short foray into the empirical process theory that underlies the main results of the
paper. First, in Section~\ref{sec:VC-theory} we compute the Vapnik-Chervonenkis (VC) dimension of a class of monotonically-indexed functions. In Section~\ref{sec:proof-main}, 
we prove the main results, Theorems~\ref{thm:nonasymp_risk} and~\ref{thm:fclt}, and sketch the 
proof of the technical Lemma $\ref{lem:bootvalid}$ regarding the bootstrap. Careful proofs for all results can be found in Appendix~\ref{app:proofs}.

\subsection{Monotonically-indexed function classes} \label{sec:VC-theory}

Because our results hold even without reference to the risk control problem studied earlier in the paper, we adopt new notation that reflects the underlying empirical process that we are tasked with controlling.  

Let $\cZ$ be a probability space on which we define i.i.d.\ random variables $Z_1, \dots, Z_n$. 
We begin by introducing the following notion of monotonically-indexed function classes.

\paragraph{Monotonically-indexed function class:} 
We say the collection of functions $\cF = \{f_t \colon \cZ \to \RR\}_{t \in [0, 1]^K}$ is 
\emph{monotonically-indexed} if 
\[
t \preceq s \quad \mbox{implies} \quad 
f_t(z) \leq f_s(z),~\mbox{for any}~z\in \cZ.
\]
Here, the inequality $t \preceq s$ corresponds to the sequence of componentwise inequalities $t_i \leq s_i$ for $i = 1, \ldots, K$.
It does not matter what the space $\cZ$ is over which the functions $f_t$ are defined, as the monotonicity provides enough structure to restrict the complexity of the class. 

A monotonically-indexed class $\cF$ should not be confused with a class $\cF$ of monotone functions, which is the class such that whenever $f \in \cF$, then $z_1 \preceq z_2$ implies $f(z_1) \leq f(z_2)$. This latter setting has been studied classically; see, for instance, \cite{dehardtGeneralizationsGlivenkoCantelliTheorem1971}.

Our first result shows that monotonically-indexed function classes have small VC subgraph dimension, as defined in~\cite{vandervaartAsymptoticStatistics1998}. 

\begin{prop}\label{prop:vc} If $\cF$ is monotonically-indexed, then its VC (subgraph) dimension is at most $K + 1$. 
\end{prop}
\noindent See Appendix~\ref{app:lempropproof} for a proof of this claim. Our proof is inspired by Lemma 9.10 in~\cite{kosorokIntroductionEmpiricalProcesses2008a}, which provides a proof in the case $K = 1$.

Finiteness of the VC dimension allows us to easily derive the following consequence, which is a functional central limit theorem (CLT) for monotonically-indexed classes. To state the result, we consider 
the following rescaled and centered process: 
\begin{equation}\label{eqn:centered-process}
\bG_n(t) \defn \frac{1}{\sqrt{n}} \Big(\sum_{i=1}^n (f_t(Z_i) - \EE f_t(Z_i)) \Big).  
\end{equation}
\begin{prop}\label{prop:converge} 
Let $\cF$ denote a monotonically-indexed function class. If the elements of $\cF$ are right-continuous and uniformly bounded, in the sense that
\[
\sup_{f \in \cF} \sup_{z \in \cZ} |f(z)| < \infty,
\]
then the process $\bG_n$ converges in distribution to a centered Gaussian process $\bG$ with covariance kernel 
\[ 
C(t,s) = \EE[f_t(Z) f_s(Z)] - \EE[f_t(Z)]\EE[f_s(Z)].
\]
\end{prop}
\noindent 
This result is almost an immediate consequence of our Proposition~\ref{prop:vc}, Theorem 19.14, and Lemma 19.15 of \cite{vandervaartAsymptoticStatistics1998}; for details, 
see Appendix \ref{app:lempropproof}. 

\paragraph{Comparing general monotone functions to indicators:} One interesting interpretation of 
Proposition~\ref{prop:converge} arises in the case $K = 1$. Let $\cF$ be a monotonically-indexed function 
class; let $\bG$ and $C$ denote the limit process and kernel associated with $\bG_n$ as guaranteed by 
Proposition~\ref{prop:converge}. Additionally, suppose for simplicity that $F(t) = \EE f_t(Z)$ also satisfies 
$F(0) = 0$ and $F(1) = 1$ (the general case can be reduced this case by translation and rescaling). 
Then, we consider a hypothetical process where we sample $Z'_i$ in an i.i.d.\ fashion according to the CDF $F$. We consider the following process: 
\[
\bG_n'(t) \defn \frac{1}{\sqrt{n}} \Big(\sum_{i=1}^n (\mathbf{1}\{Z'_i \leq t\} - F(t))\Big). 
\]
By Proposition~\ref{prop:converge} we have that $\bG_n'$ converges to a centered 
Gaussian process with covariance kernel $C'(t, s) = F(t \wedge s) - F(t) F(s)$. Interestingly, 
we also see that the limit processes $\bG, \bG'$ and their covariance kernels $C, C'$ are related via 
\[
C'(t, s) = C(t, s) + \Delta(t, s) , \quad 
\mbox{where}~\Delta(t, s) = \EE [f( t \wedge s, Z) - f_t(Z) f_s(Z)].
\]
It is straightforward to check that $\Delta$ is a positive semidefinite kernel on $[0, 1] \times [0,1]$, and we obtain
\[
\bG' = \bG + \mathbb{W}, 
\]
where $\mathbb{W}$ is a centered Gaussian process on $[0, 1]$
with the covariance $\Delta$.
Asymptotically, this relation shows that the empirical process associated with a monotone function is 
stochastically no larger than that of a corresponding CDF. 

\subsection{Proofs of main results}\label{sec:proof-main}

We now present the proofs of our main results, Theorems~\ref{thm:nonasymp_risk} and~\ref{thm:fclt}, 
and sketch the proof of Lemma \ref{lem:bootvalid} on validity of the bootstrap.  

\subsubsection{Proof of Theorem~\ref{thm:nonasymp_risk}}
\label{sec:proof-thm-nonasymp}

In Appendix~\ref{app:bdkwproof}, we prove the following result for function classes which are monotonically-indexed by a single parameter, which can be seen as a 
generalization of the DKW inequality.

\begin{theorem}\label{thm:dkw-os} Let $\cF = (f_t)_{t \in [0,1]}$ denote a monotonically-indexed function class, with $K = 1$.
Then, the rescaled and centered process
\[
\bG_n(t) \defn \frac{1}{\sqrt{n}} \Big(\sum_{i=1}^n (f_t(Z_i) - \EE f_t(Z_i)) \Big)
\]
satisfies 
\begin{equation}\label{eq:dkw-os}
\twomax{\PP\Big\{ \sup_t \bG_n (t) > x\Big\}}{\PP\Big\{ \sup_t -\bG_n(t) > x\Big\}}
\leq \e \, \exp(-2x^2),
\end{equation}
for all $x > 0$. 
\end{theorem}
Theorem~\ref{thm:nonasymp_risk} now follows by taking $\lambda = x$, and identifying
\[
Z_i = (X_i, Y_i) \quad  \mbox{and} \quad 
f_t(Z_i) = \ell(\cC_t(X_i), Y_i), \quad \mbox{for}~i=1,\ldots,n. 
\]
Note that the monotonicity assumption in Section \ref{sec:setting-technical} (decreasing)
is the reverse of the one for monotonically-indexed classes (increasing).

\subsubsection{Proof of Theorem~\ref{thm:fclt}}
We first identify $Z_i = (X_i, Y_i)$ and $f_t(Z_i) = \ell(\cC_t(X_i), Y_i)$, as in the proof of Theorem \ref{thm:nonasymp_risk}. Then Theorem \ref{thm:fclt} follows from 
Proposition \ref{prop:converge}, because for any $z$, right continuity in $t$
can be assumed without loss of generality by 
redefining $f_t(z)$ at the discontinuity points, which are countably many 
due to monotonicity, and uniform boundedness holds because $0 \leq f_t(z) \leq 1$. 

\subsubsection{Proof sketch of Lemma~\ref{lem:bootvalid}}

let $\cD_n = Z_1, \dots, Z_n$ be i.i.d, and let the centered, rescaled bootstrap distribution be
\[\bG_n^\star(t) = \frac{1}{\sqrt{n}}\left(\sum_{i = 1}^n (M_{n, i} - 1) f_t(Z_i) \right) \]
where $M_{n} \sim \textup{Multinomial}(n, 1/n, \dots, 1/n)$. 

The functional CLT in Theorem \ref{thm:fclt} holds if and only if the 
bootstrap distribution is accurate, in the sense that the conditional law of $\bG_n^\star \mid \cD_n$ converges to the distribution of $\bG$, in probability. The conclusion follows by a version of the continuous mapping theorem that holds for bootstrap distributions. See Appendix \ref{app:lempropproof} for details.

\section{Discussion}\label{sec:discussion}

We have presented an approach to distribution-free predictive inference that allows post hoc optimization of bounded, monotone risk functions.
Unlike most existing work, our bounds are \textit{uniform}, enabling exploratory revision of levels after seeing the data. 

An alternative would be to set $\alpha$ ahead of time via data splitting. Neither approach is better than the other in general~\cite{kuchibhotlaPostSelectionInference2022c}. 
But in the machine-learning problems that motivate our work, where uncertainty quantification may be one component of an overall complex engineering system, the property of uniform bounds seems particularly useful. It allows the choice of $t$ to be a complex function of other components of the system.

Our best performing methods, based on the bootstrap, have asymptotic validity but do not have provable finite-sample validity. This is a familiar issue---the inequality of Theorem \ref{thm:nonasymp_risk} is akin to a Hoeffding inequality, whereas the bootstrap is akin to a central limit theorem. This begs the question: it possible to derive an analog of a Bernstein inequality---a tight, variance-aware, finite-sample bound? Such a result was recently demonstrated for binary losses~\cite{bartl2023variance}, and we defer further investigation to figure work.

Another natural question concerns the extension of the present tools to confidence bounds for truly non-monotone risks, rather than near monotone risks or combinations of them. The method of Learn Then Test \cite{angelopoulosLearnThenTest2022b} achieves finite-sample risk control in the binary setting by gridding the space of parameters into $p$ points and performing tests at each point to assess whether the risk is below a level $\alpha$, with multiplicity correction. That work, however, does not estimate the underlying dependencies between the tests, a task at which bootstrap methods succeed asymptotically.

\printbibliography

\appendix

\section{Experimental details}\label{app:mscocoinfo}

\subsection{A concentration inequality}\label{app:wsr_bound}

The following bound is drawn from the work of \cite{waudby-smithEstimatingMeansBounded2023}:

\begin{equation}\label{eq:wsr-def}
    \Lhatp(t) = \inf\left\{p \ge 0: \max_{i = 1,\ldots, n}\mathcal{K}_i(p; t) > \frac{1}{\delta}\right\}.
\end{equation}
where $\mathcal{K}$ is referred to as a capital process in $i$, defined in terms of further quantities:
\begin{gather}
    \mathcal{K}_i(p; t) = \prod_{j=1}^{i}\left\{1 - \lambda_j(t) (\ell_j(t) - p)\right\}, \text{ where} \\
    \hat{\mu}_i(t) = \frac{1/2 + \sum_{j=1}^{i} \ell_{j}(t)}{1 + i}, \,\, \hat{\sigma}_i^2(t) = \frac{1/4 + \sum_{j=1}^{i}(\ell_{j}(t) - \hat{\mu}_{j}(t))^2}{1 + i}, \,\, \lambda_{i}(t) = \min\left\{1, \sqrt{\frac{2\log(1 / \delta)}{n\hat{\sigma}_{i-1}^2(t)}}\right\}.
\end{gather}

\begin{theorem}[Based on Theorem 3 of \cite{waudby-smithEstimatingMeansBounded2023}]\label{thm:wsr}
For any parameter $t$ and any finite sample size $n$, 
\begin{equation}
    \PP(L(t) \leq \Lhatp(t)) \geq 1 - \delta.
\end{equation}
\end{theorem}

As $L(t)$ is simply the mean of a bounded random variable, the same conclusion 
can hold for simpler upper bounds $\Lhatp(t)$, such as the Hoeffding bound, but
typically Waudby-Smith and Ramdas' bound seems tighter 
than any other bound with proven finite-sample validity. 

In the present setting, however, this bound is only pointwise valid, not uniform; that is,
it is not simultaneously valid for all $t \in [0,1]$, 
or any other subset $\cT \subset [0,1]$ which is not a singleton. 

\subsection{More details for MS COCO}\label{app:details_violation}

In this section, we provide additional detail regarding how we computed the quantities discussed in Section \ref{sec:violation} and Section \ref{sec:experiments} for the MS COCO dataset.

\textbf{Splitting the MS-COCO dataset.} The 2014 MS COCO dataset \cite{linMicrosoftCOCOCommon2015} consists of about 200K labeled images, of which ${\sim}120$K are designated as either training or validation images. The label of each image $X \in \cX$ is a vector $Y \in \{0,1\}^{80}$, corresponding to which of 80 classes are present in the image. We separated the 120K train/val images into three splits.

\textbf{Split 1: Training a classifer.}  Half of the images went to a split for training a TResnet model \cite{ridnikTResNetHighPerformance2020} for 29 epochs, which computes logits given $X$; we obtained a model of scores $f: \cX \to [0,1]^K$ by converting the raw logits using a sigmoid activation, and a classifier $\cC_t$ via equation \eqref{eq:threshold-classifier}. The model was cached for every epoch. 

\textbf{Split 2: Choosing the best epoch.} A second split of 1K images was used to choose from the epochs the best performing classifier in terms of mean average precision (mAP), namely epoch 5. We obtain a final classifier $\cC_t$.

\textbf{Split 3: Calculating performance metrics.} Denote the remaining third split of ${\sim}60K$ images from the train/val set as $\Dcoco$. The MS COCO image/label pairs $(X,Y)$ can be thought of as samples from some joint distribution of MS COCO-type images. But we do not know this distribution and cannot sample from it; in particular, we do not know ground truth risks, such as the FNR, from which to exactly calculate performance metrics such as miscoverage. 

(Some notation: Let $\hat L(t; \cD)$ denote an empirical risk calculated using data $\cD$, and similarly let $\Lhatp(t; \cD)$ denote some upper bound. Let $\cD^*_n$ denote a sample of size $n$ with replacement from $\cD$, and let $\texttt{true}_n$ be an iid sample of size $n$ from the true distribution of MS COCO images.)

To compute a surrogate quantity, on each simulation we randomly split $\Dcoco$ into two halves, a holdout set $\cH$ and a sampling set $\cS$; we picked $n$ points with replacement $\cS^*_n$; using these $n$ points, we computed $\hat t$ from optimizing a trade-off (see Section \ref{app:lossrisk}), where $t$ takes values on an even grid of 500 points in $[0,1]$; and then we evaluate performance metrics by treating $\hat L(t; \cH)$ as ground truth. 

For example, we approximate anywhere miscoverage, which is
\[
    \PP\left( L(t) > \Lhatp(t; \texttt{true}_n) \text{ for all } t \in [0,1] \right),
\]
by the surrogate
\[
    \PP\left( \hat L_n(t, \cH) > \Lhatp(t; \cS^*_n) \text{ for all } t \in [0,1] \mid \Dcoco  \right),
\]
so the probability is taken over the split of $\Dcoco$ as well as the sample of size $n$ from $\cS$. The function $\hat L_n (\cdot \,;\cH)$ serves as a surrogate for the true mean $L( \cdot)$, computed using a holdout set, and the sample of $n$ points $\cS^*_n$ serves as a surrogate for sampling $n$ points iid from the true distribution of MS COCO-type images.

\subsection{MS COCO losses and risks} \label{app:lossrisk}

Consider the classifier from Split 2 of Appendix \ref{app:details_violation}. The following are multi-label classification losses used for the MS COCO dataset:

\begin{align}
    \ell_\textup{FNP}(\cC_t(X), Y) &= \frac{\sum_{k = 1}^K 1\{[\cC_t(X)]_k = 0, [Y]_k = 1\}  }{ 1 \vee \sum_{k = 1}^K 1\{[Y]_k = 1\} } \\
    \ell_\textup{FPP}(\cC_t(X), Y) &= \frac{\sum_{k = 1}^K 1\{[\cC_t(X)]_k = 1, [Y]_k = 0\}  }{ 1 \vee \sum_{k = 1}^K 1\{[Y]_k = 0\} } \\
    \ell_\textup{FDP}(\cC_t(X), Y) &= \frac{\sum_{k = 1}^K 1\{[\cC_t(X)]_k = 1, [Y]_k = 0\}  }{ 1 \vee \sum_{k = 1}^K 1\{[\cC_t(X)]_k = 1\} } \\
    \ell_\textup{SetSize}(\cC_t(X), Y) &= \frac{\sum_{k = 1}^K 1\{[\cC_t(X)]_k = 1\}  }{ K },
\end{align}
and the resulting risks are as follows:
\begin{align}
\textup{FNR}(t) &= \EE[\ell_\textup{FNP}(\cC_t(X), Y)] \\
\textup{FPR}(t) &= \EE[\ell_\textup{FPP}(\cC_t(X), Y)] \\
\textup{FDR}(t) &= \EE[\ell_\textup{FDP}(\cC_t(X), Y)] \\
\textup{SetSize}(t) &= \EE[\ell_\textup{SetSize}(\cC_t(X), Y)].
\end{align}

\begin{figure}
    \centering
    \includegraphics[width = \textwidth]{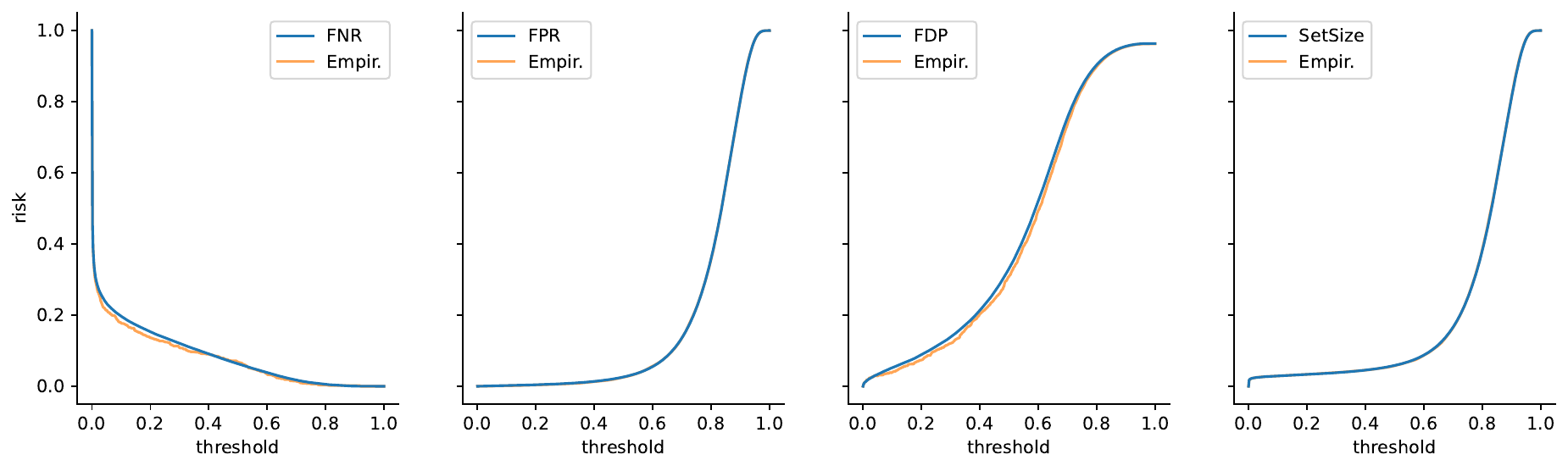}
    \caption{Four different risks on the MS COCO distribution, and an empirical estimate of them based on $n = 300$ data points.} \label{fig:mscoco_risks_appendix}
\end{figure}

An illustration of these risks is displayed in Figure \ref{fig:mscoco_risks_appendix}. Technically, the ground truth plotted in blue is actually computed based on a holdout dataset $\cH$, and the empirical estimate from sampling $n = 300$ points from a disjoint dataset $\cS$; see Appendix \ref{app:details_violation} for details. 

\paragraph{Risk tradeoffs.} All the risks are provably monotone, except for FDR, which appears to be nearly monotone. 
Since FNR is the only one which is decreasing, it makes sense
to trade all the others off of FNR. 

In particular, let $\hat t = \argmin_{t \in \That} \varphi(\hat L_n(t), \hat Q_n(t))$, where $\hat L_n$, $\hat Q_n$ represent given empirical risks, and $\varphi$ aggregates them in some way. The constraint $\That = \{t: \hat L_n(t) \leq r \}$ models the idea that $L_n$
too large would be intolerable; we set $r = 0.1$
in the experiments. Usually, this constraint was non-binding. 

We changed $\varphi, \hat L,$ and $\hat Q$ depending on the experiment.  Specifically, Figure \ref{fig:conservatism_mscoco} computes the tradeoff conservatism $\EE[\Lhatp(\hat t) - L( \hat t)]$ 
for different choices of $\varphi, \hat L,$ and $\hat Q$, 
depending on the risk being controlled:

\begin{itemize}
    \item FNR control: $\hat L$ is empirical FNR, $\hat Q$ is empirical FPR, and $\varphi(\ell, q) = \ell + q$. 
    \item FPR control: $\hat L$ is empirical FPR, $\hat Q$ is empirical FNR, and $\varphi(\ell, q) = \ell + q$. 
    \item FDR control: $\hat L$ is empirical FDR, $\hat Q$ is empirical FNR, and $\varphi(\ell, q) = - \textup{dist}\Big((\ell, q), \texttt{line}((1,0), (0,1))\Big)$. 
    \item SetSize control: $\hat L$ is empirical SetSize, $\hat Q$ is empirical FNR, and $\varphi(\ell, q) = - \textup{dist}\Big((\ell, q), \texttt{line}((1,0), (0,1))\Big)$. 
\end{itemize}

Here, $\textup{dist}((\ell,q), \texttt{line}(A,B))$ refers to the distance between the point $(\ell, q) \in \mathbb{R}^2$ and 
the closest point that lies on the line between $A, B$. 
Maximizing this distance 
finds the ``elbow'' on an ROC type curve
$(L(t), Q(t))$.

\section{Proofs}\label{app:proofs}

\subsection{Proofs of Lemmas and Propositions}\label{app:lempropproof}

To prove Lemma \ref{lem:bootvalid} we will draw upon two textbook theorems, 
written in the notation of \Cref{sec:bounds}. In particular, 
let $\cD_n = Z_1, \dots, Z_n$ be i.i.d. with common distribution $P$, and define
\[\bG_n^\star(t) = \frac{1}{\sqrt{n}}\left(\sum_{i = 1}^n (M_{n, i} - 1) f_t(Z_i) \right), \]
where $M_{n} \sim \textup{Multinomial}(n, 1/n, \dots, 1/n)$. 

The first is a central limit theorem for bootstrap processes. 

\begin{theorem}[Theorem 2.6 (i, ii) from \cite{kosorokIntroductionEmpiricalProcesses2008a}]\label{thm:kosorok_CLT}
A function class $\cF$ is $P$-Donsker if and only if
\[
    \bG_n^\star \to \bG
\]
in distribution (conditionally on $\cD_n$, in probability), and also $\bG_n^\star$ is asymptotically measurable.
\end{theorem}

Next, following the notation of the textbook \cite{kosorokIntroductionEmpiricalProcesses2008a}, 
let the Banach spaces $\mathbb{D} = \ell^\infty(\cF)$ and 
$\mathbb{E} = \mathbb{R}$ have 
the uniform norm. We have the following continuous mapping theorem 
for bootstrap processes. 

\begin{theorem}[Theorem 10.8 from \cite{kosorokIntroductionEmpiricalProcesses2008a}]\label{thm:kosorok_ctsmap}
Let $g: \mathbb{D} \to \mathbb{E}$ be continuous, and assume that the map taking
$M_n \mapsto h(\bG_n^\star)$ is measurable for every bounded, continuous $h: \mathbb{D} \to \mathbb{R}$. Then if $\bG_n^\star$ converges in distribution 
to a tight process $\bG$ (conditionally on $\cD_n$, in probability), then $g(\bG_n^\star)$ converges in distribution to
$g(\bG)$ (conditionally on $\cD_n$, in probability). 
\end{theorem}

\begin{proof}[Proof of Lemma \ref{lem:bootvalid}]

First, a minor note. This lemma was stated in the notation of Section
\ref{sec:method}; it remains true as stated using 
the notation of \Cref{sec:bounds}, 
and it is easy to translate between the two. 

Now the fact that $\bG_n$ converges in distribution to a Gaussian process $\bG$ is
the definition of $\cF$ being a $P$-Donsker class. Hence we may use Theorem \ref{thm:kosorok_CLT} to claim that $\bG_n^\star$ converges in distribution 
to the Gaussian process $\bG$ (conditionally on $\cD_n$, in probability). Now we may plug 
this into Theorem \ref{thm:kosorok_ctsmap} 
using $g(F) = \sup_{t \in [0,1]} \pm F(t)$, 
after the continuity and measurability conditions are checked. But the continuity 
follows from the continuity of the uniform norm, and the measurability 
certainly holds, so the conclusion of Theorem \ref{thm:kosorok_ctsmap} holds, which gives the result.
\end{proof}

\begin{proof}[Proof of Proposition \ref{prop:vc}]
Consider a set of distinct points $P = \{p_1, \dots, p_{K+1} \} \in \cZ \times \RR$, and let $S_t = \{(z, s) : f_t(z) < s \}$ denote a subgraph. Essentially, we perform
the classical 
calculation of the VC dimension of 
half-intervals in $\RR^K$. 

Suppose
w.l.o.g that for every $k$, $p_k \in S_t$ for some $t$, 
and define $t^*_{j, k} = \inf\{r : \exists t \text{ s.t. } t_j = r, p_k \in S_t\}$, the smallest $t_j$ such that $S_t$ picks out $p_k$.
Define $t^{**}_j = \max_{k \in [K]} t^*_{j, k}$.
We now consider two cases.

If there is $k'$ such that 
$t^*_{j, k'} < t^{**}_j$ for all $j$, 
then whenever $P \setminus \{p_{k'}\} \subset S_t$, 
then $t_j \geq t_j^{**}$ 
for each $j$. This
implies that $t_j > t^*_{j, k'}$ for all $j$, so by monotonicity, 
$P \subset S_t$, so $S_t$ 
cannot pick out the configuration $P \setminus \{p_{k'}\}$. 

If instead for 
every $k$, we have
$t^*_{j, k} = t^{**}_j$ for some $j$, then
since $j \in \{1, \dots, K \}$ and
$k \in \{1, \dots, K + 1 \}$,
there must exist $k', k''$ such that
$t^*_{j, k'} = t^{**}_{j, k''} = t_j^{**}$.
In this case, the two-point set
$\{p_{k'}, p_{k''}\}$ cannot be shattered, 
because if there is $t$ satisfying
$p_{k'} \in S_{t}$, $p_{k''} \notin S_{t}$, then
\[
    t^{**}_j = t^*_{j, k'} \leq  t_{j}
     < t^*_{j, k''} = t^{**}_j,
\]
which is a contradiction. 

\end{proof}

\begin{proof}[Proof of Proposition \ref{prop:converge}]\label{app:zrnicproof}
From Theorem 19.14 and Lemma 19.15 of \cite{vandervaartAsymptoticStatistics1998}, if a ``suitably measurable'' function class has finite VC dimension and uniformly bounded, then it is $P$-Donsker. The definition of $P$-Donsker is precisely the convergence given in the result. 

Suitable measurability is a technical condition that \cite{vandervaartAsymptoticStatistics1998} does not define, but he does note that it is sufficient that there is a countable collection of functions $\cG$ in which for each $f_t(\cdot)$ we can find a sequence $g^1, g^2, \ldots : \cZ \to \RR$ satisfying $\lim_{m \to \infty} g^m(z) \to f_t(z)$ at each $z$. By the right continuity of $f_t$, it suffices to take $\cG$ equal to the subset of $\cF$ where $t$ is rational.
\end{proof}

\subsection{Proof of Theorem \ref{thm:locsim_bootstrap_risk}}

The following proof uses the same core idea as that of Theorem 1 of 
\cite{zrnicLocallySimultaneousInference2023}, except that we extend their original,
nonasymptotic examples to the present asymptotic setting of the bootstrap. 
Also, we use the notation of Section \ref{sec:method} rather than their 
general notation, which enables a less abstract proof.

In addition to the notation of Section \ref{sec:method}, we set additional 
notation for the proofs and two lemmas. 
Let the realized set of selected $t$ be denoted as
\[
    \That = \{t : \hat L_n(t) \leq r \},
\]
let the population sublevel set be denoted as
\[
    \trueT = \{t : \hat L_n(t) \leq r \},
\]
let the predictive set of selected $t$, with known mean, be denoted as
\[
    \predT = \left\{t : L(t) \leq r + \frac{\hatqglob}{\sqrt{n}} \right\},
\]
and let the predictive set of selected $t$, with estimated mean, be denoted as
\[
    \predThat = \left\{t : \hat L_n(t) \leq r + 2\frac{\hatqglob}{\sqrt{n}} \right\}.
\]
Though we do not directly use this fact, the predictive sets have the 
property that when an independent copy of $\hat L_n$ is observed, and hence
an independent copy of the realized set $\That$, the predictive sets 
contain the new realized set with high probability.

For each $t$, let the one-sided and two-sided confidence sets for $L(t)$ be
\[
\Btloc{\predT} = \left\{y : y \leq \hat L_n(t) + \frac{\hatqloc{\predT}}{\sqrt{n}} 
\right\}
\]
and
\[
\Btglob = \left\{y : |y - \hat L_n(t)| \leq  \frac{\hatqglob}{\sqrt{n}} \right\}.
\]
They are confidence bands when interpreted as functions of $t$.

Finally, let $\bG$ refer to the Gaussian process of Theorem \ref{thm:fclt}; it is
the process to which the empirical process $\sqrt{n}(\hat L_n(t) - L(t))$ is convergent. 

Now we establish a lemma on the convergence of bootstrap quantiles. 
\begin{lemma}\label{lem:bootqtl}
    If $\hatqglob = \inf_q \{q : \PP(D_{n, \star} > q \mid \cD_n) 
    \leq \delta_{\mathrm{glob}}) \}$
    is the conditional $1 - \delta_{\mathrm{glob}}$ quantile of $D_{n, \star}$, 
    and for any $\cT$, $\hatqloc{\cT} = \inf_q \{q : \PP(D_{n, \star}^-(\cT) > q \mid \cD_n) 
    \leq \delta_{\mathrm{loc}}) \}$ is the conditional $1 - \delta_{\mathrm{loc}}$ 
    quantile of $D_{n, \star}^-(\cT)$, then
    \begin{equation}\label{eq:proof_cvg1}
        \hatqglob \overset{p}{\to} \qglob
    \end{equation}
    and 
    \begin{equation}\label{eq:proof_cvg2}
        \hatqloc{\predT} \overset{p}{\to} \qloc{\trueT},
    \end{equation}
    where $\qglob = \inf_q \{q : \PP(\sup_{t \in [0,1]} |\bG(t)| > q) 
    \leq \delta_{\mathrm{loc}}) \}$ is the limiting global quantile, and 
    $\qloc{\trueT} =  \inf_q \{q : \PP(\sup_{t \in 
    \trueT} \bG(t) > q) \leq \delta_{\mathrm{loc}})\}$ 
    is the limiting one-sided quantile.
\end{lemma}

\begin{proof}
    First, equation \eqref{eq:proof_cvg1} holds as an immediate consequence of
    Lemma \ref{lem:bootvalid}. Next, note that $D_{n, \star}^-(\predT)$ can be 
    expressed as
    \[
        D_{n, \star}^-(\predT) = \sup_{t \in [0,1]} \bG_n^\star(t) \cdot 1\{L(t) \leq r + \frac{\hatqglob}{\sqrt{n}} \} 
    \]
    By Lemma \ref{lem:bootvalid}, the conditional law of $\bG_n^\star$ converges
    to that of $\bG$ in probability, and the process
    $1\{L(t) \leq r + \frac{\hatqglob}{\sqrt{n}}\}$ converges to the constant 
    function $1\{L(t) \leq r\} = 1\{t \in \trueT \} $ in probability, 
    so by Slutsky's lemma, the conditional law of 
    \[
        \bG_n^\star(t) \cdot 1\{L(t) \leq r + \frac{\hatqglob}{\sqrt{n}} \} 
    \]
    converges to the law of
    \[
        \bG(t) \cdot 1\{t \in \trueT \} 
    \]
    in probability. Finally, by a similar application of Theorem \ref{thm:kosorok_ctsmap} 
    as in the proof of 
    Lemma \ref{lem:bootvalid}, the conditional law of
    $D_{n, \star}^-(\predT)$ converges to that of $\sup_{t \in 
    \trueT} \bG(t)$ in probability. This directly implies
    equation \eqref{eq:proof_cvg2}.
\end{proof}

Second, we establish the key lemma of the proof, which mirrors Lemma 1 
of \cite{zrnicLocallySimultaneousInference2023}. 

\begin{lemma}\label{lem:keylem}
For any $n$, the inequality
\begin{equation}
\begin{split}
    \PP(L(t) \in \Btloc{\predThat}\, \textup{for all } t \in \That)  \geq 
    \PP\Big( \{L(t) &\in \Btloc{\predT}\, \textup{for all } t \in \predT\} 
    \\ &\textsc{ and }  \{L(t) \in \Btglob\, \textup{for all } t \in [0,1]\}\Big)
\end{split}
\end{equation}
holds.
\end{lemma}
The two events in the right-hand side have been bracketed for clarity. 
Note that these events depend on the
sample size $n$, but this dependence has been suppressed.
\begin{proof}
Because $\Btloc{\cT_1} \subset \Btloc{\cT_2}$ when $\cT_1 \subset \cT_2$, it is sufficient to show the inclusion 
\[
   \That \subset \predT \subset \predThat
\]
on the event $L(t) \in \Btglob\, \textup{for all } t \in [0,1]$. 

On this event, the first inclusion holds because 
$L(t) \leq  \hat L_n(t) + \hatqglob / \sqrt{n}$ for all $t \in [0,1]$, 
giving the chain of implications
\[
\hat L_n(t) \leq r \Rightarrow \hat L_n(t) + \hatqglob / \sqrt{n} 
\leq r +  \hatqglob / \sqrt{n} 
\Rightarrow L(t) \leq r +  \hatqglob / \sqrt{n}.
\]
The second inclusion holds because 
$\hat L_n(t) \leq  \hat L_n(t) + \hatqglob / \sqrt{n}$ for all $t \in [0,1]$,
giving the chain of implications 
\[
L(t) \leq r + \hatqglob /\sqrt{n} \Rightarrow L(t) + \hatqglob / \sqrt{n} 
\leq r +  2\hatqglob / \sqrt{n} 
\Rightarrow \hat L_n(t) \leq r +  2\hatqglob / \sqrt{n}.\qedhere
\]
\end{proof}

\begin{proof}[Proof of Theorem \ref{thm:locsim_bootstrap_risk}]

We wish to show that the event
\[
\left\{L(t) \leq \hat L_n(t) + \frac{\hatqloc{\predThat}}{\sqrt{n}}, \quad \mbox{simultaneously for all}~t \in \widehat{\mathcal{T}}_{r} \right\}
\]
occurs with probability at least $1 - (\delta + o(1))$. Observe that we may
re-express this event as the event
\[
\left\{L(t) \in \Btloc{\predThat},\quad \textup{for all } t \in \That\right\}.
\]

Applying Lemma \ref{lem:keylem}, we can lower bound the probability of this event
in terms of two events with easier-to-handle sets $\cT$:
\begin{equation}
\begin{split}
    \PP(L(t) \in \Btloc{\predThat}\, \textup{for all } t \in \That)  \geq 
    \PP\Big( \{L(t) &\in \Btloc{\predT}\, \textup{for all } t \in \predT\} 
    \\ &\textsc{ and }  \{L(t) \in \Btglob\, \textup{for all } t \in [0,1]\}\Big)
\end{split}
\end{equation}

Rewrite the right-hand side as
\begin{equation}
\begin{split}
    \PP\Big( \{\sqrt{n}(L(t) - \hat L_n(t)) \leq \hatqloc{\predT} \, \textup{for all } t \in \predT\} 
    \textsc{ and }  \{\sqrt{n} |L(t) - \hat L_n(t)| \leq \hatqglob \, \textup{for all } t \in [0,1]\}\Big),
\end{split}
\end{equation}
which can be further re-expressed as
\begin{equation}
    \PP\Big( \{\sup_{t \in [0,1]} \big(\sqrt{n}(L(t) - \hat L_n(t) \big) \cdot  1\{L(t) \leq r + \frac{\hatqglob}{\sqrt{n}}\} \leq \hatqloc{\predT} \}
    \textsc{ and }  \{\sup_{t \in [0,1]} \sqrt{n} |L(t) - \hat L_n(t)| \leq \hatqglob \}\Big),
\end{equation}
and by the union bound, this is greater than 
\begin{equation}
\begin{split}
    1 - \PP\Big( \sup_{t \in [0,1]} \big(\sqrt{n}(L(t) - \hat L_n(t) \big) \cdot {1}\{L(t) \leq r + \frac{\hatqglob}{\sqrt{n}}\} > \hatqloc{\predT} \Big) 
    - \PP\Big( \sup_{t \in [0,1]} \sqrt{n} |L(t) - \hat L_n(t)| > \hatqglob\Big).
\end{split}
\end{equation}
Now applying the functional
central limit theorem (Theorem \ref{thm:fclt}), Lemma \ref{lem:bootqtl}, and Slutsky's
lemma, this equals
\[
    1 - \PP\Big( \sup_{t \in [0,1]} \bG(t) \cdot {1}\{t \in \trueT\} \leq \qloc{\trueT}\Big) \\
    - \PP\Big( \sup_{t \in [0,1]} |\bG(t)|  \leq \qglob \Big) - o(1),
\]
where we have assumed each $\hat q$ and $q$ represent exact 
quantiles as in Lemma \ref{lem:bootqtl}; this is fine to assume because this quantity
lower bounds the case where they are not exact quantiles.

Importantly, $\trueT$ is not a random set, 
but is deterministic, so because $\qloc{\trueT}$ and
$\qglob$ are quantiles, definitionally we have that the previous display is lower
bounded by
\[
1 - \delta_\mathrm{loc} - \delta_\mathrm{glob} - o(1) = 1  - (\delta + o(1)),
\]
as claimed.

\end{proof}

\subsection{Proof of Theorem \ref{thm:dkw-os}}\label{app:bdkwproof}

To prove this theorem we need two lemmas.

\begin{lemma}\label{lem:disperse}
Let $(Y(t))_{t \in \RR}$ be a real-valued stochastic process with bounded sample paths. Let $\varphi: \RR \to \RR$ be a convex, non-decreasing function. Let $V = \sup_t \EE[Y(t) \mid Z]$ and $V^* = \sup_t Y(t)$. Then
\[ 
    \EE[\varphi(V)] \leq \EE[\varphi(V^*)].
\]
\end{lemma}{}

\begin{proof}[Proof of Lemma \ref{lem:disperse}]
First, note that for any function $f: \RR \to \RR$ with $\sup_t f(t) < \infty$, then if $\varphi: \RR \to \RR$ is monotone increasing and continuous, 
\[
    \varphi(\sup_t f(t)) = \sup_t \varphi(f(t)).
\]
The direction $\varphi(\sup_t f(t)) \geq \sup_t \varphi(f(t))$ follows by using monotonicity. The direction $\varphi(\sup_t f(t)) \leq \sup_t \varphi(f(t))$ follows from continuity: $\varphi(\tilde f) \leq \sup_t \varphi(f(t))$ for any $\tilde f$ in the range of $f$, so take an increasing sequence $\tilde f_1, \tilde f_2, \dots$ converging to $\sup_t f(t)$. 

We get the chain of inequalities, noting that $\varphi$ must be continuous due to convexity:
\begin{gather*}
    \EE[\varphi( \sup_t \EE[Y(t) \mid Z] ) ] = \EE[ \sup_t \varphi (\EE[Y(t) \mid Z] ) ] \leq \EE[ \sup_t   \EE[\varphi (Y(t)) \mid Z]  ] \\ 
    \leq \EE[ \EE[ \sup_t\varphi (Y(t)) \mid Z]  ]
    = \EE[ \sup_t\varphi (Y(t)) ] = \EE[ \varphi ( \sup_tY(t)) ],
\end{gather*}
where the first inequality is Jensen's.
\end{proof}

The next lemma concerns what we call the ``Bentkus transform,'' defined in \cite{bentkusDominationTailProbabilities2006a}. For any function $S: \RR \to \RR$, define the log-concave hull $S^\circ$ as the smallest function such that $S \leq S^\circ$ and $x \to -\log S^\circ(x)$ is a convex function. If $S$ is a survival function, define its Bentkus transform as
  \[
    \mathcal{B}[S](x) = \inf_{r < x} \frac{\EE[(X - r)_+]}{(x - r)_+} = \inf_{r < x} \frac{1}{(x - r)_+} \int_r^\infty S(y) \, dy, 
  \]
where $X \sim 1 - S$. 

\begin{lemma}\label{lem:kemperman}
For a survival function $S$, for all $x \in \RR$,
\[
    S(x) \overset{\text{(i)}}{\leq} \mathcal{B}[S](x) \overset{\text{(ii)}}{\leq} e S^\circ(x).
\]
\end{lemma}
\begin{proof}

Inequality (i) can be shown by Markov's inequality, applied to the random variable $(X - r)_+$. Inequality (ii) is proved in more generality as Lemma 4.2 by \cite{bentkusHoeffdingInequalities2004}, or alternately Lemma 1.1 \cite{bentkusDominationTailProbabilities2006a}. 
\end{proof}

\cite{bentkusHoeffdingInequalities2004} attributes this lemma to \cite{pinelisFractionalSumsIntegrals1999b}, and a special case to Kemperman, citing Ch. 25 of \cite{shorackEmpiricalProcessesApplications2009}. It seems to be well-suited for proving extremal results for random variables that are the ``least averaged.'' 

For instance, it was used to prove Theorem 1.2 of \cite{bentkusHoeffdingInequalities2004}, which showed that binary random variables are, in some sense, more stochastic than variables bounded in $[0,1]$. Meanwhile, Ch. 25 of \cite{shorackEmpiricalProcessesApplications2009} demonstrates a DKW inequality for independent but not identically distributed random variables by showing that the iid case is the most stochastic. Evidently, these results are related to ours. 

\begin{proof}[Proof of Theorem \ref{thm:dkw-os}]

Let $h(t, Z) := f_t(Z)$, and we extend its domain to $t \in \RR$ by taking $h(t, Z) = 1$ whenever $t > 1$, and $h(t, Z) = 0$ whenever $t < 0$; then
\[
    \sup_{t \in [0,1]} H_n(t, Z) - H(t) = \sup_{t \in \RR} H_n(t, Z) - H(t),
\]
so from now on we take suprema over $\RR$, and show
\[
    \PP\Big(\sup_{t \in \RR} + \big(H_n(t, Z) - \EE H_n(t, Z) \big) \geq x\Big) \leq e \exp(-2x^2),
\]
which implies the result with the $+$ sign (the $-$ sign is exactly similar). 

The right-continuity assumption implies that, for each $Z$, $h(t,Z)$ is a CDF of a random variable supported on $[0,1]$; that is, it is non-decreasing, right-continuous, and $h(1, Z) = 1$. Then, conditionally on each $Z_i$, let $T_i$ be a random variable drawn according to the CDF $h(\cdot, Z_i)$, and define
\[
Y(t) = \sqrt{n}\left( \frac{1}{n} \sum_{i = 1}^n 1\{T_i \leq t\} - H(t) \right).
\]
Letting $Z = (Z_1, \dots, Z_n)$, it follows that
\[
\EE[Y(t) \mid Z] = \sqrt{n}\left( \frac{1}{n} \sum_{i = 1}^n h(t, Z_i) - H(t) \right) = \sqrt{n} (H_n(t) - H(t)).
\]
Let  $V = \sup_t \EE[Y(t) \mid Z]$ and $V^* = \sup_t Y(t)$. Set $\varphi_r: x \mapsto (x - r)_+$, an increasing convex function, and applying Lemma \ref{lem:disperse}, we have for any $x > r$
\[
    \frac{\EE[\varphi_r(V)]}{\varphi_r(x)} \leq \frac{\EE[\varphi_r(V^*)]}{\varphi_r(x)}.
\]
Let $S$ denote the survival function of $V$ and $S^*$ of $V^*$. Then taking infimums in $r$ on both sides, we can write an inequality between two Bentkus transforms:
\[
    \mathcal{B}[S](x) \leq \mathcal{B}[S^*](x).
\]
Applying Lemma \ref{lem:kemperman} gives
\begin{equation}\label{eq:icx_result}
    S(x) \leq e[{S^*}]^\circ(x),
\end{equation}
and finally, observe that the (one-sided) DKW inequality \cite{massartTightConstantDvoretzkyKieferWolfowitz1990} implies $S^*(x) \leq \exp(-2x^2)$ . Since the right-hand side is log-concave, in fact $[{S^*}]^\circ(x) \leq \exp(-2x^2)$. So ultimately
\[
    S(x) = \PP\Big(\sup_{t \in \RR} + \big(H_n(t, Z) - \EE H_n(t, Z) \Big) \leq e \exp(-2x^2),
\]
as claimed.
\end{proof}

Inspecting the argument leading up to equation \eqref{eq:icx_result}, we can extract a fact that may be of independent interest; a weaker version was 
also stated by \cite{panchenko2003symmetrization}. It concerns the \textit{increasing convex ordering} of random variables (see, e.g., \cite{shakedUnivariateMonotoneConvex2007}, Section 4.A).  

We write $A \leq_{\textup{icx}} B$, read as ``$A$ is less than $B$ in the increasing convex order,'' if $\EE \varphi(A) \leq \EE \varphi(B)$ for all non-decreasing, convex $\varphi$. Let $S_A, S_B$ denote their survival functions.
\begin{cor}
Whenever $A \leq_{\textup{icx}} B$, then $S_A(x) \leq e S_B^\circ(x).$
\end{cor}

\end{document}